\documentclass[a4paper, reqno]{amsart}

\usepackage{amsmath,amssymb,microtype}
\usepackage{enumerate}
\usepackage{hyperref}
\usepackage{url}
\usepackage{tabto}
\usepackage{multirow}
\usepackage{tikz}
\usepackage[textsize=footnotesize]{todonotes}
\usetikzlibrary{automata,positioning,shapes}

\theoremstyle{plain}
\newtheorem{theorem}{Theorem}

\newtheorem{lemma}[theorem]{Lemma}
\newtheorem{corollary}[theorem]{Corollary}
\newtheorem{claim}[theorem]{Claim}

\newcommand{\A}{\mathcal{A}}

\newcommand{\R}{\mathsf{R}}

\newcommand{\F}{\mathsf{F}}
\newcommand{\V}{\mathsf{V}}
\newcommand{\oF}{\mathsf{1F}}
\newcommand{\oV}{\mathsf{1V}}
\newcommand{\tF}{\mathsf{2F}}
\newcommand{\tV}{\mathsf{2V}}

\newcommand{\sign}{\ensuremath{\textsf{sign}}}

\newcommand{\REG}{\mathsf{Reg}}
\newcommand{\VPL}{\mathsf{VPL}}
\newcommand{\DOCL}{\mathsf{DOCL}}
\newcommand{\DCFL}{\mathsf{DCFL}}
\newcommand{\Len}{\mathsf{Len}}
\newcommand{\SL}{\mathsf{SL}}

\newcommand{\LI}{\mathsf{LI}}
\newcommand{\Fin}{\mathsf{Fin}}
\newcommand{\Triv}{\mathsf{Triv}}
\newcommand{\classA}{\mathsf{A}}

\newcommand{\rest}{\mathord\restriction}

\newcommand{\bigO}{\ensuremath{{\mathcal O}}}
\newcommand{\calA}{\ensuremath{{\mathcal A}}}
\newcommand{\calB}{\ensuremath{{\mathcal B}}}

\newcommand{\rpush}[1][a]{\ensuremath{\textsf{rightpush}(#1)}}
\newcommand{\lpush}[1][a]{\ensuremath{\textsf{leftpush}(#1)}}
\newcommand{\rpop}{\ensuremath{\textsf{rightpop}()}}
\newcommand{\lpop}{\ensuremath{\textsf{leftpop}()}}
\newcommand{\query}{\ensuremath{\textsf{query}()}}

\newcommand{\eps}{\varepsilon}
\newcommand{\Conf}{\mathrm{Conf}}

\begin{document}

\title{Low-Latency Sliding Window Algorithms for Formal Languages}

\author[M.~Ganardi]{Moses Ganardi}
\email{ganardi@mpi-sws.org}
\author[L.~Jachiet]{Louis Jachiet}
\email{louis.jachiet@telecom-paris.fr}
\author[M.~Lohrey]{Markus Lohrey}
\email{lohrey@eti.uni-siegen.de}
\author[T.~Schwentick]{Thomas Schwentick}
\email{thomas.schwentick@tu-dortmund.de}

\address[Moses Ganardi]{Max Planck Institute for Software Systems, Kaiserslautern, Germany}
\address[Louis Jachiet]{LTCI, T\'{e}l\'{e}com Paris, Institut Polytechnique de Paris, France}
\address[Markus~Lohrey]{Universit{\"a}t Siegen, Germany}
\address[Thomas Schwentick]{TU Dortmund University, Germany}
\thanks{Markus Lohrey has been partially supported by the DFG research project LO 748/13-1.}

\begin{abstract}
Low-latency sliding window algorithms for regular and context-free languages are studied, where latency
refers to the worst-case time spent for a single window update or query. For every regular language $L$ it is shown that there
exists a constant-latency solution that supports adding and
removing symbols independently on both ends of the window
(the so-called two-way variable-size model). We prove that this result extends to all visibly pushdown languages. For deterministic
1-counter languages we present a $\bigO(\log n)$ latency sliding window algorithm for the two-way variable-size model
where $n$ refers to the window size.
We complement these results with a conditional lower bound:  there exists a fixed real-time deterministic context-free language $L$ such that,
assuming the OMV (online matrix vector multiplication) conjecture,  there is no
sliding window algorithm for $L$ with latency $n^{1/2-\epsilon}$ for any $\epsilon>0$, even in the
most restricted sliding window model (one-way fixed-size model). The above mentioned results all refer to the unit-cost RAM model with 
logarithmic word size. For regular languages we also present a refined picture using word sizes $\mathcal{O}(1)$,
$\mathcal{O}(\log\log n)$, and $\mathcal{O}(\log n)$.
\end{abstract}

\maketitle

\section{Introduction}

\subparagraph{Sliding window algorithms}
In this paper, we investigate sliding window algorithms for formal
languages. In the basic sliding window model, an infinite stream $s = a_1 a_2 a_3 \dots$ of symbols from a finite alphabet $\Sigma$
is read symbol by symbol from left to right.
It works \emph{one-way} and with a {\em fixed window size} $n$. The
{\em window content} is the suffix of length $n$ of the prefix of 
the stream $s$ seen so far. Thus, in each step, a new right-most
symbol is read into the window and the left-most symbol is moved out. A sliding
window algorithm for a language $L \subseteq \Sigma^*$ has to indicate
at every time instant whether the current window content belongs to
$L$ (initially the window is filled with some dummy symbol). The two resources
that one typically tries to minimize are memory and the worst-case time spent per incoming symbol.
It is important to note that the model only has access to the letter
currently read. In particular, if the algorithm wants to know  
the precise content of the current window or, in particular, which letter moves out of
the window, it has to dedicate memory for it. For general background
on sliding window algorithms see \cite{Aggarwal07,DatarGIM02}.

We refer to the above sliding window model as the {\em one-way, fixed-size model}. A more general variant is the one-way, \emph{variable-size}
sliding window model. In this model the arrival of new symbols and the expiration 
of old symbols are handled independently, i.e., there are update operations
that add a new right-most symbol and an operation that removes the
left-most symbol. Therefore the size of the window can 
grow and shrink. This allows to model for instance a time-based window that contains
all data values that have arrived in the last $t$ seconds for some fixed $t$. If the 
arrival times are arbitrary then the window size may vary.
The \emph{two-way model} is a further generalization whose update
operations allow to add and remove symbols on both
sides of the window. It can be combined with both the fixed-size and
the variable-size model. The variable-size two-way model is the most general model; it is also known
as a deque (double-ended queue); see \cite[Section 2.2.1]{Knuth97}.
An algorithm also needs to handle query operations, asking whether the
current window content is in the language $L$.

There are two important complexity measures for a sliding window algorithm: its
{\em space complexity} and its {\em latency} (or {\em time complexity}), i.e.\ the time required for a single update or query operation. 
Both are usually expressed depending on the window size $n$ (for a fixed-size
sliding window algorithm) or the maximal window size $n$ that occurs during a run of the sliding window algorithm
(for a variable-size sliding window algorithm).
 In this paper we are mainly interested in the {\em latency} of sliding window algorithms.
 Since it turns out that space complexity is an important tool for proving lower bounds on the latency
 of sliding window algorithms, we first discuss known results on space
 complexity.

\subparagraph{Space complexity of sliding window algorithms}
The space complexity of sliding window algorithms for formal languages
in the one-way (fixed-size and variable-size) model has been studied in \cite{Ganardi19,GHKLM18,GHL16,GanardiHL18,GanardiJL18}
and in \cite[Section~9]{GanardiHL17arxiv} for the two-way variable-size model.
For regular languages, the main result of \cite{GHL16} is a space
trichotomy for the one-way case: the space complexity of a regular language is either
constant, logarithmic or linear. This result holds for  
the fixed-size model as well as the variable-size model, although the
respective language classes differ
slightly. For the two-way variable-size model a space trichotomy has been shown in \cite{GanardiHL17arxiv}.
Table~\ref{table:results} summarizes some of the main
results of \cite{GHKLM18,GHL16,GanardiHL17arxiv} in more detail 
(ignore the word size bound for the moment).
The results on the two-way fixed-size model in Table~\ref{table:results} are shown in this paper.
In that table,
\begin{itemize}
\item $\REG$ denotes the class of all regular languages;
\item $\Len$ denotes the class of \emph{regular length languages}, i.e.,
 regular languages $L \subseteq \Sigma^*$, for which either $\Sigma^n
 \subseteq L$ or $\Sigma^n \cap L = \emptyset$, for every $n$;
\item $\LI$ denotes the class of all \emph{regular left ideals}, i.e., regular languages of the form
   $\Sigma^* L$ for a regular language $L$;
\item $\SL$ denotes the class of \emph{suffix  languages}\footnote{In
    \cite{GHKLM18,GHL16}, we used instead of $\SL$ the boolean closure of $\SL$ (the so-called
    suffix testable language); but this makes no difference, since we are only interested in the boolean closures
    of language classes.}, i.e., languages of
  the form $\Sigma^* w$;
\item $\Triv$ denotes class of trivial languages, i.e., the class consisting of $\emptyset$
and $\Sigma^*$ only;
\item $\Fin$ denotes the class of finite languages;
\item  $\langle \classA_1,\ldots,\classA_n \rangle$ denotes the
  Boolean closure of $\bigcup_{1 \le i \le n} \classA_i$.  
\end{itemize}
Note that these classes are defined with respect to an alphabet,
e.g.\ $a^*$ is trivial if the alphabet is $\{a\}$ but non-trivial if the alphabet is $\{a,b\}$.
In Table~\ref{table:results}, we write $f \in \bar{\Theta}(g)$  for $f,g : \mathbb{N} \to \mathbb{R}$ iff
$f \in \mathcal{O}(g)$ and there is a constant $c>0$ with $f(n) \geq c \cdot g(n)$ for infinitely many $n$.

Some of the  results from \cite{GHKLM18,GHL16} 
were extended to (subclasses of) context-free languages in \cite{Ganardi19,GanardiJL18}. 
A space trichotomy was shown for visibly pushdown languages in \cite{Ganardi19}, whereas for the class
of all deterministic context-free languages the space trichotomy fails \cite{GanardiJL18}.

\subparagraph{Content of the paper}
In this paper we consider the  latency of sliding window
algorithms for regular and deterministic context-free languages in all four of
the above models: one-way and two-way, fixed-size and variable-size.  These models are formally defined in Section~\ref{sec-SW}.
As the algorithmic model, we use  the standard RAM model.  The word size (register length) is a parameter in this
model  and we allow it to depend on the fixed window size $n$ (in the fixed-size model) or the maximal window size $n$
that has occurred in the past (for the variable-size model). More precisely, depending on the language class, the word
size can be $\bigO(1)$ (resulting in the {\em bit-cost model}), $\bigO(\log \log
n)$, or $\bigO(\log n)$. We assume the unit-cost measure, charging a cost of 1 for each basic register
operation.

The bit-cost model serves as a link to transfer lower bounds: it is a simple observation, formalized in Lemma~\ref{lemma-S-T}, that
the sliding window time complexity for a language
$L$ in the bit-cost model is at least the logarithm of the minimum possible space
complexity. In fact, this lower bound holds even with respect to the
non-uniform bit-probe model, where we have a separate algorithm for
each window size $n$. And, again by Lemma~\ref{lemma-S-T}, lower bounds on the latency in the bit-cost model
translate to lower bounds on the word size for unit-cost algorithms with constant latency.

\renewcommand{\arraystretch}{1.3}
\begin{table}  \small
\begin{tabular}{rl|c|c|c|c}
 & & $\oF$ & $\oV$ & $\tF$ & $\tV$  \\ \hline
 word size: & \!\!\!\!\!$\bigO(1)$   & \multirow{2}{*}{$\langle \Len, \SL\rangle$} & \multirow{2}{*}{$\Triv$} & \multirow{2}{*}{$\langle \Len, \Fin\rangle$} & \multirow{2}{*}{$\Triv$}  \\
 space in bits: & \!\!\!\!\!$\bigO(1)$ & & & & \\ \hline
  word size: & \!\!\!\!\!$\bar{\Theta}(\log\log n)$ & \multirow{2}{*}{$\langle \Len, \LI\rangle \setminus \langle \Len, \SL\rangle$} & \multirow{2}{*}{$\langle \Len, \LI\rangle \setminus\Triv$} & \multirow{2}{*}{$\emptyset$} & \multirow{2}{*}{$\Len\setminus\Triv$} \\ 
 space in bits: & \!\!\!\!\!$\bar{\Theta}(\log n)$ & & & & \\ \hline
  word size: & \!\!\!\!\!$\bar{\Theta}(\log n)$ & \multirow{2}{*}{$\REG \setminus \langle \Len, \LI\rangle$} & \multirow{2}{*}{$\REG \setminus\langle \Len, \LI\rangle$} & \multirow{2}{*}{$\REG \setminus  \langle \Len, \Fin\rangle$} & \multirow{2}{*}{$\REG  \setminus \Len$} \\
 space in bits: & \!\!\!\!\!$\bar{\Theta}(n)$ & & & &
\end{tabular}

\bigskip

\caption{Summary of  results for regular languages and constant latency. The columns correspond to the 4 different sliding window models
 ($\oF$ = one-way fixed-size, $\oV$ = one-way variable-size, $\tF$ = two-way fixed-size, $\tV$ = two-way variable-size).
The rows correspond to different combinations of word size and space in bits. For all three combinations the latency is $\mathcal{O}(1)$.
Note that the language classes in each column yield a partition of $\REG$.}
\label{table:results}
\end{table}

In Section~\ref{sec:reg}, we study the latency of sliding window algorithms for regular languages and offer
a complete picture.  Our  contribution here is mainly of algorithmic
nature, since most of the lower bounds are
simple consequences of space lower bounds, that were shown in \cite{GHKLM18,GHL16,GanardiHL17arxiv}. The main result of the first part is that these lower bounds 
can be achieved by concrete algorithms. More precisely, there are algorithms that (1) achieve the
optimal latency with respect to the bit-cost model and (2) constant latency with respect to unit-cost model, and (3) also have optimal space
complexity. The precise results are summarized in 
Table~\ref{table:results} for the unit-cost model. In all cases, the sliding window algorithms have constant latency.
For example, languages from $\langle \Len, \LI\rangle$ have one-way sliding window algorithms with constant latency
on unit-cost RAMs with word size
$\bigO(\log \log n)$ and thus $\bigO(\log \log n)$ latency in
the bit-cost model. These algorithms have space complexity $\bigO(\log
n)$. Moreover, unless the language belongs to $\langle \Len, \LI\rangle$ (for the fixed-size model)
or $\Triv$ (for the variable-size model) these resource bound cannot be improved.
Note that, while for three of the four models there is a
trichotomy,  the two-way fixed-size model is an outlier: it has a
dichotomy, since there are no languages
of intermediate complexity.

In Section~\ref{sec-cfl} we considers the latency for deterministic context-free languages ($\DCFL$) and here
the study is more of an explorative nature. Since every $\DCFL$ has a linear time
parsing algorithm~\cite{KNUTH1965607}, one might hope to get also a low-latency sliding window algorithm.
Our first result tempers this hope: assuming 
the OMV (online matrix vector multiplication) conjecture \cite{HenzingerKNS15}, we show that 
there exists a fixed real-time $\DCFL$ $L$ such
that no algorithm can solve the
sliding window problem for $L$ on a RAM with logarithmic word
size with latency $n^{1/2-\epsilon}$ for any $\epsilon>0$, even in the
one-way fixed-size model (the most restricted model). This motivates to look for subclasses
that allow more efficient sliding window algorithms. We present two results in this direction.
We show that for every visibly pushdown language ($\VPL$) \cite{AlurM04} there is a two-way variable-size sliding window algorithm on a unit-cost
RAM with word size $\bigO(\log n)$ and constant latency. Visibly pushdown languages are widely used, e.g. for describing tree-structured documents and 
traces of recursive programs. They share many of the nice algorithmic and closure properties of regular languages. Finally, we show that
for every deterministic one-counter language ($\DOCL$) there is a two-way variable-size sliding window algorithm on a unit-cost
RAM with word size $\bigO(\log n)$ and latency $\bigO(\log n)$.

\subparagraph{Related work}
The latency of regular languages in the sliding window model has been first studied in
\cite{TangwongsanH017}, where it was shown that in the one-way, fixed-size model,
every regular language has a constant latency algorithm on a RAM with word size $\log n$
(the result is not explicitly stated in  \cite{TangwongsanH017} but directly
follows by using the main result of \cite{TangwongsanH017} for the
transformation monoid of an automaton). Our upper bound results for
general regular languages rely on this work and we extend its
techniques to visibly pushdown languages.

A sliding window algorithm can be viewed as a dynamic data structure that
maintains a dynamic string $w$ (the window content) under very restricted update operations. 
Dynamic membership problems for more general
updates that allow to change the symbol at an arbitrary position have been studied 
in \cite{AmarilliJP21,FrandsenHMRS95,FrandsenMS97}.

Standard streaming algorithms (where the whole history and not only
the last $n$ symbols is relevant) for
visibly pushdown languages (and subclasses) were studied in \cite{BabuLRV13,BabuLV10,BathieS21,FischerMS18,FrancoisMRS16,KrebsLS11,MagniezMN14}.
These papers investigate the space complexity of streaming. Update times of streaming algorithms for timed automata have 
been studied in \cite{grez21}.

\section{Sliding window model} \label{sec-SW}

Throughout this paper we use $\log n$ as an abbreviation for $\lceil \log_2 n \rceil$.

Consider a function $f : \Sigma^* \to C$ for some finite alphabet $\Sigma$ and some countable set $C$.
We will view the sliding window problem for the function $f$
as a dynamic data structure problem, where we want to
maintain a word $w \in \Sigma^*$, called the {\em window}, which undergoes changes
and admits membership queries to $L$. Altogether, we consider the
following operations on $\Sigma^*$, where for a word $w = a_1\cdots a_n \in \Sigma^*$ we write $|w| = n$ for its length
and $w[i:j]=a_i\cdots a_j$ for the factor from position $i$ to position $j$ (which is $\varepsilon$ if $i > j$).
\begin{itemize}
\item \rpush: \tabto{6em} Replace $w$ by $w a$.
\item \lpush: \tabto{6em} Replace $w$ by $a w$.
\item \lpop: \tabto{6em} Replace $w$ by $w[2:|w|]$ (which is $\varepsilon$ if $w = \varepsilon$).
\item \rpop: \tabto{6em} Replace $w$ by $w[1:|w|-1]$ (which again is $\varepsilon$ if $w = \varepsilon$).
\item \query: \tabto{6em} Return the value $f(w)$.
\end{itemize}
In most cases, the function $f$ will be the characteristic function of a language $L \subseteq \Sigma^*$; in this 
case we speak of the sliding window problem for the language $L$.

In the \emph{two-way model} all five operations are allowed, whereas  in the \emph{one-way model}, we only allow to add symbols on the right and to
remove symbols on the left, that is, it allows only the operations
\rpush, \lpop, and \query. 
In the \emph{variable-size window model} the operations can be applied
in arbitrary order, but in the \emph{fixed-size window model},  push
operations always need to be followed directly by a pop operation on the
other side. More formally, each $\rpush$ needs to be immediately
followed by a $\lpop$ and (in the two-way
 model), each $\lpush$ needs to be
immediately followed by a $\rpop$. In particular, no query can occur
between a $\lpush$ and the subsequent $\rpop$. Therefore, as the name suggests, in
the fixed-size model the string always has the same length $n$, for
some $n$, if we consider a push and its successive pop operation as
one operation. 

In the variable-size model the window $w$ is initially $\varepsilon$,
whereas in the fixed-size model it is initialized as $w = \Box^n$ for
some default symbol $\Box \in \Sigma$. We
allow algorithms a preprocessing phase and disregard the time they
spend during this initialization.
In the fixed-size model the algorithm  receives the window size $n$
for its initialization.

We denote the four combinations of models by $\oF$, $\oV$, $\tF$, and $\tV$,
where 1 and 2 refer to one-way and two-way, respectively, and $\F$ and $\V$
to fixed-size and variable-size respectively.

We use two different computational models to present our results, the
uniform \emph{word RAM model} for upper bounds and the
non-uniform \emph{cell probe model} for lower bounds. 

\subparagraph{Word RAM model}
We present algorithmic results (i.e., upper bounds) in the {\em word RAM model} with
maximal word size (or register length) of $b(n)$ bits, for some
\emph{word size function} $b(n)$.
Algorithms may also use registers of length smaller than $b(n)$, for a more
fine-grained analysis.
In the fixed-size model $n$ is the fixed window size, whereas
in the variable-size model $n$ is the maximum window size that has appeared in the past.
In particular, if the window size increases then also the allowed word size
$b(n)$ increases,
whereas a subsequent reduction of the window size does not decrease the
allowed word size. As usual, all RAM-operations on registers of word size at most $b(n)$ take constant time (unit-cost assumption).
For a model $M \in \{ \oF, \oV, \tF, \tV \}$, an $M$-algorithm (for a function $f$)
is a sliding window algorithm that supports the operations of model
$M$.

An $M$-algorithm $\calA$ has {\em latency} (or {\em time complexity}) $T(n)$ if for all $n$ the following hold:
\begin{itemize}
\item If $M \in \{ \oF, \tF\}$, then in every computation of window size $n$,
all operations of model $M$ are handled within $T(n)$ steps by $\calA$.
\item If $M \in \{ \oV, \tV\}$, then in every computation of maximal window size $n$,
all operations of model $M$ are handled within $T(n)$ steps by $\calA$.
\end{itemize}
\emph{Space complexity} is defined accordingly and refers to the number of bits used by the algorithm.

\subparagraph{Cell probe model} For lower bounds we use the {\em cell probe model}.
We formalize the model only for  sliding
window algorithms. For a model $M \in \{ \oF, \oV, \tF,
\tV \}$, an $M$-algorithm in the cell probe model is a collection  $\calA = (\calA_n)_{n \geq 0}$, where
$\calA_n$ is an $M$-algorithm for window size $n$ (if $M = \oF$ or $M = \tF$), respectively, maximal window size $n$ (if $M = \oV$ or $M = \tV$). 
Furthermore, we only count the number of read/write accesses to memory cells and disregard computation completely. We also say that 
$\calA = (\calA_n)_{n \geq 0}$ is a {\em non-uniform $M$-algorithm}.

More formally, fix a word size function $b = b(n)$.
In the cell probe model an $M$-algorithm $\calA_n$ for (maximal) window length $n$ is a collection of
decision trees $t_{n,\textsf{op}}$ for every operation  $\textsf{op}$ of model $M$.
Each node of $t_{n,\textsf{op}}$ is labelled with a register operation $\mathsf{read}(R_i)$ or $\mathsf{write}(R_i,u)$
where $i$ is a register address and $u \in \{0,1\}^{b_i}$. Here,
$b_i \le b(n)$ denotes the size (in bits) of register $i$. 
A node labelled with $\mathsf{read}(R_i)$ has $2^{b_i}$ children, one
for each possible value of register $i$.
A node labelled with $\mathsf{write}(R_i,w)$ has exactly one child.
Moreover, the leaves of the decision tree for \textsf{query}$()$ are labelled with output values ($0$ or $1$).
The latency $T^{\calA}(n)$ of $\calA$ is the maximal height of a decision tree $t_{n,\textsf{op}}$.
The space complexity $S^{\calA}(n)$ of $\calA$ is the sum over the bit
lengths of the different registers referenced in all trees $t_{n,\textsf{op}}$.

By minimizing for every $n$ the number of
bits used in the trees $t_{n,\textsf{op}}$, it follows that for every
language $L \subseteq \Sigma^*$ 
there is a (non-uniform) $M$-algorithm
$\calB$ with optimal space complexity $S^{\calB}(n)$ for every $n$. We
denote this optimal space complexity by $S^M_{L}(n)$; see also
\cite{GHKLM18}. Since for every language $L \subseteq \Sigma^*$ 
there is a non-uniform $M$-algorithm that stores the window explicitly with $n
\cdot \log |\Sigma|$ bits, it holds $S^M_{L}(n) \leq n \cdot
\log|\Sigma|$.

Space complexity in the sliding window model was analyzed in~\cite{Ganardi19,GHKLM18,GHL16,GanardiHL18,GanardiJL18}  
for the one-sided models ($\oF$ and $\oV$) and \cite[Section~9]{GanardiHL17arxiv} for the model $\tV$ (the model $\tF$ has not been studied so far).
In these papers, the space complexity was defined slightly different but equivalent to our definition.
Note that lower bounds (for space and time) that are proved for the
cell probe model also hold for the (uniform) RAM model.

As mentioned before, we will use the non-uniform cell probe model only for lower bounds.
Lower bounds on the space complexity of sliding window algorithms
yield lower bounds on the latency, as well. In fact, all our
(unconditional) lower bounds stem from space lower bounds with the
help of the following lemma. We mainly apply space lower bounds 
 from \cite{GHKLM18,GHL16,GanardiHL17arxiv}.

\begin{lemma} \label{lemma-S-T}
  For each model $M \in \{ \oF, \oV, \tF, \tV\}$ and each non-uniform $M$-algorithm $\calA$ with word size 
  $b(n)$ for some language $L$, it holds $b(n) \cdot T^{\calA}(n) \ge \log S^M_L(n) - \mathcal{O}(1)$.
\end{lemma}

\begin{proof}
We first show the statement for $b(n)=1$.
Let $\calA$ be a non-uniform algorithm with the properties from Lemma~\ref{lemma-S-T} with $b(n)=1$. 
Hence, each decision tree for (maximum) window size $n$ has height at most $t_n := T^{\calA}(n)$ and hence at most $2^{t_n}$
many nodes (since $b(n)=1$, every node has at most two children).
Since every node refers to only one register, we can bound the number of registers by $2^{t_n}$
and hence $S^M_L(n) \le (2|\Sigma|+2) \cdot 2^{t_n}$. Here,
$2|\Sigma|+2$ accounts for the maximal number of possible operations
of any model $M$. The claim of the lemma follows.

For the general case of an arbitrary word size function $b(n)$, note that
every non-uniform algorithm $\calA$ with word size $b(n)$ and time bound
$T^{\calA}(n)$ induces an algorithm $\calA'$ with word size 1 and time
bound at most $b(n) \cdot T^{\calA}(n)$, we immediately get the statement of the lemma.
\end{proof}
For a new space lower bound, we use the following fooling set approach.

\begin{lemma} \label{lemma-fooling-set}
  Let $M$ be a model, $L \subseteq \Sigma^*$ a language, and $U \subseteq \Sigma^n$ a set of strings of length $n$.
  If for all $u,v\in U$ with $u \neq v$ there
  exists a sequence $\alpha$ of operations of model $M$ such that
  $\alpha(u)\in L \Leftrightarrow \alpha(v)\not\in L$ and for every prefix of $\beta$ of $\alpha$ the lengths of $\beta(u)$ and $\beta(v)$
  are bounded by $n$, then $S^M_L(n)\ge \log |U|$.
\end{lemma}

\begin{proof}
The proof is straightforward. If we assume that the optimal algorithm for (maximum) window size $n$
uses fewer than $\log |U|$ bits, then there must exist two different strings $u,v \in U$ such that the memory
states of the algorithm after pushing $u$ and $v$ into the sliding window are the same. Clearly, after applying 
the sequence of operations $\alpha$, the memory states are still the same, so the algorithm cannot tell the
difference concerning $L$-membership of $u$ and $v$, yielding a contradiction.
\end{proof}
We call a set $U$ fulfilling the property from Lemma~\ref{lemma-fooling-set} \emph{fooling set} (for $L$ and window size $n$).

\section{Regular languages}\label{sec:reg}

As mentioned in the introduction, the class of regular languages
satisfies a space trichotomy in the models $\oF$ and $\oV$~\cite{GHKLM18,GHL16}: a regular language either has space
complexity $\bar{\Theta}(n)$ or $\bar{\Theta}(\log n)$ or $\bigO(1)$.
In the light of Lemma~\ref{lemma-S-T}, the best
we can therefore hope for are constant latency sliding-window algorithms
with word size $\bigO(\log n)$, $\bigO(\log \log n)$ and $\bigO(1)$,
respectively. It turns out that such algorithms actually exist.
In the following, we consider each of these three levels separately.
We present algorithms and confirm their optimality by corresponding lower bounds.

Before we start, let us fix our (standard) notation for finite automata.
A {\em deterministic finite automaton} (DFA) is a tuple $\A = (Q,\Sigma,q_0,\delta,F)$
where $Q$ is a finite set of states, $\Sigma$ is an alphabet, $q_0 \in Q$ is the initial state,
$\delta \colon Q \times \Sigma \to Q$ is the transition function and $F \subseteq Q$ is the set of final states.
The transition function $\delta$ is extended to a function $\delta \colon Q \times \Sigma^* \to Q$ in the usual way.
The language accepted by $\A$ is denoted by $L(\A)$.

\subsection{Logarithmic word size} \label{sec-daba}

For the upper bound, we show that regular languages have constant latency
sliding-window algorithms with logarithmic word size and optimal space
complexity in the two-way variable-size model and thus in all four models. 
In other words: we give a constant time implementation of a deque data structure that allows to query the content of a deque with respect to a regular language.

To this end, we start from a known $\oV$-algorithm for evaluating products over finite monoids
and adapt it so that it also works for the two-way model and meets our optimality requirements.
Recall that a \emph{monoid} is a set $\mathcal{M}$ equipped with an associative binary operation on $\mathcal{M}$.
Let $\mathsf{prod}_\mathcal{M} \colon \mathcal{M}^* \to \mathcal{M}$ be the function which maps a word over $\mathcal{M}$
to its product.
There is a folklore simulation of a queue (aka.~$\oV$-sliding window) by two stacks, which takes constant time per operation on \emph{average}, see~\cite{TangwongsanH017}.
This idea can be turned into a $\oV$-algorithm for $\mathsf{prod}_\mathcal{M}$, if $\mathcal{M}$ is a finite monoid,
taking constant time on average and $\mathcal{O}(n)$ space.
Tangwongsan, Hirzel, and Schneider presented a \emph{worst-case} constant latency algorithm~\cite{TangwongsanH017}.

\begin{theorem}[c.f.~\cite{TangwongsanH017}] \label{thm-cell-probe-old}
Let $\mathcal{M}$ be a fixed finite monoid (it is not part of the input). Then there is a $\oV$-algorithm for $\mathsf{prod}_\mathcal{M}$
with word size $\bigO(\log n)$ and latency $\mathcal{O}(1)$.
\end{theorem}
An immediate corollary of Theorem~\ref{thm-cell-probe-old} is that 
every regular language $L$ has a $\oV$-algorithm
with word size $\bigO(\log n)$ and latency $\mathcal{O}(1)$. For this, one 
takes for the monoid $\mathcal{M}$ in Theorem~\ref{thm-cell-probe-old} the transformation monoid of a DFA for $L$. 
The transformation monoid of a DFA with state set $Q$ 
and transition function $\delta : Q \times \Sigma \to Q$ is the submonoid of $Q^Q$ (the set
of all mappings on $Q$) generated by the functions 
$q \mapsto \delta(q,a)$, where $a \in \Sigma$. The monoid operation is the composition of functions:
for $f,g \in Q^Q$ we define $fg \in Q^Q$ by $(fg)(q) = g(f(q))$ for all $q \in Q$.

In order to obtain for every regular language a $\tV$-algorithm
with word size $\bigO(\log n)$, latency $\bigO(1)$, and space complexity $\bigO(n)$,
we strengthen Theorem~\ref{thm-cell-probe-old}:
 
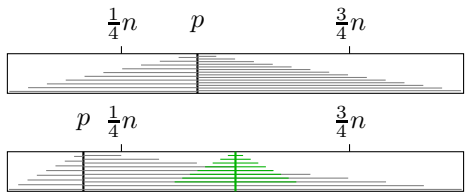
\begin{figure}[t]
\begin{tikzpicture}[scale=0.5]

\draw (0,-0.05) rectangle ++(12,1.05);
\draw (3,1) -- (3,1.2);
\draw (9,1) -- (9,1.2);
\node at (3,1.8) {$\frac{1}{4}n$};
\node at (9,1.8) {$\frac{3}{4}n$};
\node at (5,1.8) {$p$};
\draw[thick] (5,-0.05) -- (5,1);

\foreach \x in {1,...,10}
{
\draw[gray] (5,1-0.1*\x) -- (5-0.495*\x,1-0.1*\x);
}
\foreach \x in {1,...,14}
{
\draw[gray] (5,1-0.07*\x) -- (5+4.95*0.1*\x,1-0.07*\x);
}

\end{tikzpicture}
\hspace{1cm}
\begin{tikzpicture}[scale=0.5]

\draw (0,-0.05) rectangle ++(12,1.05);
\draw (3,1) -- (3,1.2);
\draw (9,1) -- (9,1.2);
\node at (3,1.8) {$\frac{1}{4}n$};
\node at (9,1.8) {$\frac{3}{4}n$};

\node at (2,1.8) {$p$};
\draw[thick] (2,-0.05) -- (2,1);
\foreach \x in {1,...,8}
{
\draw[gray] (2,1-0.125*\x) -- (2-0.245*\x,1-0.125*\x);
}
\foreach \x in {1,...,10}
{
\draw[gray] (2,1-0.1*\x) -- (2+0.995*\x,1-0.1*\x);
}

\draw[thick,black!30!green] (6,-0.05) -- (6,1);
\foreach \x in {1,...,8}
{
\draw[black!30!green] (6+0.2*\x,1-0.1*\x) -- (6-0.2*\x,1-0.1*\x);
}
\end{tikzpicture}
\caption{Left: A guardian at position $p$, storing all suffixes of $m_1 \dots m_{p-1}$ and all prefixes of $m_p \dots m_n$.
Right: As soon as $p$ escapes the range $[\frac{1}{4}n, \frac{3}{4}n]$ a new guardian (in green) is started at $\frac{n}{2}$.}
\label{figure:guardians}
\end{figure}

\begin{theorem} \label{thm-DABA-new} 
Let $\mathcal{M}$ be a fixed finite monoid. Then there is a $\tV$-algorithm for $\mathsf{prod}_{\mathcal{M}}$
with word size $\bigO(\log n)$, latency $\mathcal{O}(1)$, and space complexity $\mathcal{O}(n)$.
\end{theorem}
Before we prove Theorem~\ref{thm-DABA-new} we show 
a lemma that allows us to restrict to the case that the maximal window size $n$ is known at advance (i.e., during
initialization). 

\begin{lemma} \label{lemma-A}
Let $T$ be a non-decreasing function and let $A$ be a $\tV$-algorithm for a function $f$
such that, when initialized with the empty window and a
maximal window size of $n$, 
\begin{itemize}
\item $A$ works on a unit-cost RAM with word size $\bigO(\log n)$,
\item the initialization takes time $T(n)$,
\item all later window operations run in time $T(n)$ and 
\item $A$ stores in total at most $\bigO(n)$ bits.
\end{itemize}
With such an $A$, we can build a $\tV$-algorithm $A^*$ for $f$ without the maximal window
size limitation and with latency $\bigO(T(4n))$, space complexity 
 $\bigO(n)$ and word size $\bigO(\log n)$.
\end{lemma}

\begin{proof}
To get the statement of the lemma but with an amortized complexity bound, we
could design the algorithm $A^*$ to work just like any ``dynamic
array'': let us denote with $A_n$ the version of the algorithm that works for 
window size up to $n$. Assume that we currently work with $A_n$. 
If the window size grows to $n$ we switch to $A_{2n}$ and if the window size shrinks to $n/4$
we switch to $A_{n/2}$. To do the switching, we transfer the current window content 
using {\sf rightpush}-operations (we could also use {\sf leftpush}) of $A_{2n}$ or $A_{n/2}$ from the 
current version $A_n$ to the new version. We call this a reset. 

If we reset to $A_n$ we know that the window size was $n/2$. Therefore, the reset takes time 
$(n/2+1)\cdot T(n)$ (the $+1$ comes from the initialization which takes time $T(n)$).
On the other hand, the next reset only happens if the window size grows to $2n$ or shrinks to $n/4$.
Therefore, we make at least $n/4$ non-reset operations before the next reset happens. This leads
to an amortized latency of $\bigO(T(4\ell))$, where $\ell$ is the current window size. The term $T(4\ell)$ comes from the fact
that at each time instant we work with a version $A_n$, where $n$ is at most a factor 4 larger than
the actual window size.

To get non-amortized time bounds we need to maintain two instances of
$A$ at all time. One instance $A_{cur}$ solving the problem (just like
in the amortized case) and one instance $A_{nxt}$ that we prepare in
the background so that whenever we need to double or halve the maximal
window size, the instance $A_{cur}$ can just be replaced with $A_{nxt}$.

Our algorithm thus stores two instances of $A$ ($A_{cur}$ and
$A_{nxt}$), a copy of the current word $w$ in the window (as an
amortized circular buffer), some bookkeeping to know what is the size
$p$ of the prefix of $w$ that has been loaded into $A_{nxt}$, the
current size of $w$ and the parameter $n$ of the maximal window size of
$A_{cur}$. Overall the algorithm stores $\bigO(\ell)$ bits where $\ell$ is the current window size $\ell$.

To prepare $A_{nxt}$, we decide that whenever the current window size
is above $n/2$, we prepare $A_{nxt}$ to work with $2n$ and whenever
the window size is below $n/2$, we prepare $A_{nxt}$ to work with $n/2$. Note
that each time the window size crosses the $n/2$ threshold, our
algorithm will reset $A_{nxt}$ but each reset only takes time
$T(2n)$ or $T(n/2)$, hence time $T(2n)$ at most (if $T$ is monotone).

Propagating the effect of each window operation on $A^*$ to $A_{cur}$
and $w$ is easy, we just apply the effect for $w$ and call the right
operation for $A_{cur}$. For $A_{nxt}$, if the window operation is a
$\lpop$ or $\lpush$, we carry out the same operation in $A_{nxt}$.
This ensures that at any point during the algorithm, $A_{nxt}$ stores
a prefix of the current window whose length $p-1$ (case $\lpop$) or
$p+1$ (case $\lpush$). If the operation is a $\rpop$ or $\rpush$, we
simply ignore it in $A_{nxt}$.  Now, in both cases, i.e., for every
update operation, we call $\rpush$ in $A_{nxt}$ for up to 5 elements
$a$ ($w[p+1]$, $w[p+2]$, $w[p+3]$, $w[p+4]$, and $w[p+5]$) so that
overall $p$ (the length of the prefix of the window copied to
$A_{nxt}$) increases by at least 4. Note that there might be less than
$5$ elements to push if $A_{nxt}$ is almost ready.

Finally, note that because we switch between $A_{cur}$ and $A_{nxt}$ when
the window size reaches $n/4$ or $n$ and we restart the above copying process
from $A_{cur}$ to $A_{nxt}$ whenever the window size becomes larger or smaller than $n/2$,
we know that at least $n/4$ window updates must occur between the last restart and
the actual switch from $A_{cur}$ to $A_{nxt}$. Therefore by increasing 
the size of the prefix covered by $A_{nxt}$ by 4  for each window update, we know that $A_{nxt}$
covers the full window whenever we have to switch. Overall we do have
a scheme with the right latency as we only do a constant
number of calls to algorithm $A$ for each window operation.
\end{proof}

\begin{proof}[Proof of Theorem~\ref{thm-DABA-new}]
By Lemma~\ref{lemma-A} it suffices to consider the case, that the maximal window size $n$ is known at advance (i.e., during
initialization). We outline an algorithm  with word size $\bigO(\log n)$, latency $\mathcal{O}(1)$, and space complexity $\mathcal{O}(n)$.

The limit $n$ allows us to pre-allocate a circular bit array of size
$\bigO(n)$, in which the window content $w = m_1 m_2 \cdots m_\ell \in \mathcal{M}^*$ ($\ell \le n$) is stored
in a circular fashion and updated in time $\bigO(1)$. The
current window size $\ell$ is stored in a register of bit length $\log n$.
Also a constant number of
pointers into $w$, including pointers to the first and to the
last entry are stored. Furthermore, to keep track of the product $m_1 m_2 \cdots
m_\ell$ under the allowed operations, some sub-products $m_i m_{i+1} \cdots m_j$ for $i<j$ 
have to be stored as well. Storing all such products would require quadratic space.
Instead, the algorithm stores (again in a circular fashion) for some index $p \in [1,\ell]$, called the \emph{guardian}, 
all products $m_i m_{i+1} \cdots m_p$ ($i \in [1,p-1]$) and $m_p m_{p+1} \cdots
m_j$ ($j \in [p+1,\ell]$); see Figure~\ref{figure:guardians} for an illustration. If these
products are only stored for all $i \in [k,p-1]$, $j \in [p+1,\ell]$ for some $k, \ell$, then 
we speak of a \emph{partial guardian}.
As long as the guardian $p$ satisfies $1<p<n$, a push operations just
adds one product and a pop operation removes one. 
A $\lpop$ decreases $p$ and a
$\lpush$ increases it.  However, the algorithm only works correctly as long
as the guardian is strictly between 1 and $n$.

Therefore, we enforce the invariant $p \in [\frac{1}{8}\ell, \frac{7}{8}\ell]$. To
guarantee this invariant, we start a new guardian $p'$, initially set to $\frac{\ell}{2}$, whenever the old
guardian $p$ escapes the interval $[\frac{1}{4}\ell, \frac{3}{4}\ell]$.
We need to make sure that the computation of the products for the new guardian is fast enough such that
(i) $p$ stays in $[\frac{1}{8}\ell',\frac{7}{8}\ell']$
while the guardian $p'$ is partial, where $\ell'$ is the window size at that point,
and (ii) $p'$ stays in $[\frac{1}{4}\ell',\frac{3}{4}\ell']$.
A simple calculation shows that it suffices if the guardian $p'$ is complete after $\frac{1}{7}\ell$ steps
where $\ell$ is the window size when $p'$ was initialized. Indeed, in
worst case, the length of the string after $\frac{1}{7}\ell$ steps is
$\frac{6}{7}\ell$. If the guardian is initially at position
$\frac{1}{4}\ell$, it might afterwards be, again in worst case, at
position
$(\frac{1}{4}-\frac{1}{7})\ell=\frac{6}{56}\ell=\frac{1}{8}\times\frac{6}{7}\ell$. The
dual case is analogous.

In each step (application of an operation), 8 new products are computed in a
balanced\footnote{In principle, four products are computed for each
side, but if the word grows towards one side, this can be reflected in
the choice of the next products.}  fashion, so that, after at most
$\frac{\ell}{7}$ steps all products are available for the new guardian
$p'$. In fact, if $\ell$ denotes the window size when the computation of
the new guardian starts, after $\frac{\ell}{7}$ steps $\frac{8}{7}\ell$
products are computed, covering the potential window size after these steps.

The two guardians can be stored in registers of $\log n$ size and for
the two collections of partial products $\bigO(n)$ bits suffice.
Moreover, all update operations can be carried out within a constant
number of steps.
\end{proof}
Using again the transformation monoid of a regular language we obtain
from Theorem~\ref{thm-DABA-new}:

\begin{corollary} \label{thm-reg-upper}
Every regular language $L$ has a $\tV$-algorithm with word size $\bigO(\log n)$, latency
$\mathcal{O}(1)$ and space complexity $\bigO(n)$.
\end{corollary}
The lower bounds for the  $\bigO(\log  n)$-time level can be summarized
as follows. Recall that non-uniform $M$-algorithms 
refer to the cell probe model from Section~\ref{sec-SW}.

\begin{theorem} \label{thm-reg-lower}
  For a regular language $L$ and sliding-window model $M$, every non-uniform
  $M$-algorithm with word size $1$ for $L$ has 
  latency at least $\log n - \bigO(1)$ for infinitely many $n$ in each of the following  three cases:
  \begin{enumerate}
  \item[(a)] $M \in \{\oF, \oV\}$ and $L\not\in \langle \Len, \LI\rangle$,
  \item[(b)]  $M = \tV$ and $L\not\in \Len = \langle \Len\rangle$,
  \item[(c)]  $M = \tF$ and $L\not\in \langle \Len, \Fin\rangle$.
  \end{enumerate}
 In (b) and (c) the lower bound 
  $\log n - \bigO(1)$ holds for all $n$.
\end{theorem}

\begin{proof}
  In all three cases, the statement follows from Lemma~\ref{lemma-S-T}  using a linear lower bound on
  the space complexity.

  For (a), the space lower bound of $\Omega(n)$ for regular languages $L\not\in \langle \Len, \LI\rangle$
  follows from \cite[Theorem 9]{GHL16} and \cite[Theorem 5.1]{GHKLM18}. 
  Statement (b) follws from \cite[Theorem 9.6]{GanardiHL17arxiv}, which states that 
  in the two-way variable-size model every regular language $L\not\in  \Len$ needs linear space.
     
Towards (c), let $L\not\in \langle \Len, \Fin\rangle$ be a regular
language. It follows that there are infinitely many $n$, for which
there exist strings $u_n\in L$ and $v_n\not\in L$ of length $n$. When
moving from $u_n$ towards $v_n$ while flipping one symbol in each step,
there must occur strings $u'_n\in L$ and $v'_n\not\in L$ of length $n$ that differ
in only one position $j_n$. Let $a_n$ (resp., $b_n$) by the symbol at position $j_n$ in $u'_n$ (resp., $v'_n$).
We claim that for each such $n$, the set $\{a_n, b_n\}^n$ is a fooling set for $L$ and window size $n$. Indeed, let $u, v \in \{a_n,b_n\}^n$  differ in some position $i$. It is
easy to see that from $u$ and $v$ the strings $u'_n$ and $v'_n$ can be
obtained with an identical sequence $\alpha$ of two-way fixed-size
operations. Basically, position $i$ has to be moved to position $j_n$
and before and after $j_n$  the identical parts of $u'_n$ and
$v'_n$ have to be constructed. Since we found a fooling set of size $2^n$ for infinitely many
$n$, Lemma~\ref{lemma-fooling-set}  implies $S_L(n)\ge n$ for infinitely many $n$.
\end{proof}

\begin{corollary}
In each of the cases of Theorem~\ref{thm-reg-lower}, any $M$-algorithm with latency $\bigO(1)$ requires word size $\Omega(\log n)$.
\end{corollary}

\subsection{Sublogarithmic word size}

There are two combinations of subclasses of regular languages and sliding window models, for which we 
obtain constant latency algorithms with word size $\bigO(\log \log n)$:

\begin{theorem} \label{thm-loglog}
Let $L$ be a regular language and $M$ a sliding-window model. If (i) $L \in \langle\Len, \LI\rangle$ and $M = \oV$ or (ii)
$L \in \Len$ and $M = \tV$, then there exists an $M$-algorithm for $L$ with the following properties, where
$n$ is the maximum window size:
\begin{itemize}
\item The algorithm uses a RAM with a bit array of length $\mathcal{O}(\log n)$ and a constant number of pointers into the bit array. Each 
pointer can be stored in a register of length $\mathcal{O}(\log \log n)$.
\item The latency of the algorithm is $\bigO(1)$.
\end{itemize}
\end{theorem}

\begin{proof}
We first consider the case that $L \in \Len$ and $M = \tV$.
In principle, it suffices to keep track of the current window size $n$
and to check whether it is in some fixed semilinear set $S$. In fact, from
$L$ one compute a number $N$ and two finite sets $A$ and $B$ such
that a string is in $L$ if and only if its length $n$ satisfies (1) $n \geq N$ and $n \bmod N \in A$ or (2) $n < N$ and $n \in B$.

Obviously, the number $n$ can be stored with 
$\bigO(\log n)$ bits. However, $n$ needs to be incremented and
decremented and compared, and the naive way of doing that may require the manipulation of
$\log n$ bits for one operation, in worst case. Therefore we need to make use of a more sophisticated data
structure that allows updates with a constant number of
operations that manipulate only $\bigO(\log \log n)$ bits each.

Whenever $n\ge N$ holds, the data structure uses a counter $m$ that
keeps track of $n-N$, and a variable $r$ of constant size that keeps
track of $n \bmod N$ to check whether $n \bmod N \in A$. Whenever $n<N$, it
maintains 
$n$ by a variable of constant size and maintains whether $n\in B$ holds.

During initialization, $m$ is set to zero, and whenever $n<N$ it stays
at zero. The challenging part is to maintain $m$ if $n>N$ and to recognize whenever
a phase with $n>N$ ends by reaching $m=0$. This can be done using the 
counting techniques from \cite{FrandsenMS97}. 
In a nutshell, the method basically works just as incrementing
a binary number by, say, a Turing machine. However, to avoid long
delays, it encodes the position of the head of the Turing machine in the string
(by an underscore of the digit at that position) and each ``move'' of
the Turing machine corresponds to an incremented number. E.g., the
numbers 0,1,2,3,4 could be represented by the strings
$00\underline{0}, 00\underline{1}, 01\underline{0}, 01\underline{1},
0\underline{1}0$. It is not obvious how to determine the number represented by such a string. However, for our purposes it suffices to increment numbers, 
to compare them with a fixed constant number, and to initialize them with a
fixed number.

The algorithm maintains  the number $m$ by a bit string $s$ of length
$\ell$, for some $\ell$ such that $2^{\ell+1}-\ell-2 \geq m$.  The number $\ell$ is stored with $\bigO(\log \ell)$ bits.
One position $x$ in $s$ is considered as marked and stored with
$\bigO(\log \ell)$ bits as well.  The representation of the number $m$ by $s$ and $x$ is defined as follows
(where the marked bit in $s$ is underlined and $u,v$ are bit strings):
\begin{itemize}
\item The number $0$ is represented as $0^{\ell-1} \underline{0}$.
\item If number $t$ is represented by $u \underline{0}$ then $t+1$ is represented by $u \underline{1}$.
\item If number $t$ is represented by $u 0\underline{1} v$ then $t+1$ is represented by $u 1\underline{0} v$.
\item If number $t$ is represented by $u \underline{0}0 v$ then $t+1$ is represented by $u 0\underline{0} v$.
\item If number $t$ is represented by $u 1 \underline{1} v$ then $t+1$ is represented by $u \underline{1} 0v$.
\item If number $t$ is represented by $\underline{1} 0^{\ell -1}$ then
  $t+1$ is represented by $1 \underline{0} 0^{\ell -1}$, and $\ell$
  subsequently has to be incremented by 1.
\end{itemize}
The last of these cases does not appear in \cite{FrandsenMS97}; it is needed since in \cite{FrandsenMS97} $m$ and $\ell$ 
are fixed, which is not the case in our application. It is a crucial observation that all bits to the right of the marked
position are always zero. In particular this allows to identify when a
situation of the form $\underline{1} 0^{\ell -1}$ is obtained. We note
that additional leading zeros do not spoil the representation of a
number. Therefore $\ell$ needs never be decreased when the window size $n$ and hence $m$ is
decreased. 

One can show that in this way every number is represented in a unique
way, ignoring unmarked, leading zeros, see
\cite{FrandsenMS97}.
The above rules also specify how to increment and decrement $m$ (for decrement one has to reverse the rules). Note that
the rules modify $s$ only in a local way; therefore they can be 
implemented in time $\mathcal{O}(1)$ on a RAM of word size $\bigO(\log \log n)$. 

It remains to explain how to check whether $m=0$ holds.
To this end, an additional number is stored, which is the minimal
position $y$ in the bit string such that all positions to the left
of it are 0.\footnote{The positions are numbered from right to left,
  i.e., the right-most position has number 0.}  For $y$, $\bigO(\log \ell)$  bits suffice, as well. Its
manipulation is straightforward with the above rules. 

If $n$ is the maximal window size seen in the past,
the algorithm stores a bit array of length $\bigO(\log n)$ for the bit string $s$ and
three registers of bit length $\bigO(\log \log n)$ for the numbers $\ell, x, y$.
This completes the description of the case  that $L \in \Len$ and $M = \tV$.

\medskip
\noindent
We now consider the case that $L \in \langle\Len, \LI\rangle$ and $M = \oV$.
It suffices to consider the cases that $L \in \Len$ and $L \in \LI$: if we have a boolean
combination of such languages, we can run the sliding window algorithms for these languages
in parallel and combine their results according to the boolean formula. 
The case $L \in \Len$ is covered by the first part of the proof (for $M = \tV$),
so it suffices to consider the case $L \in \LI$ and $M = \oV$. 

Hence, $L$ is a left ideal. Let $\A = (Q,\Sigma,q_0,\delta,F)$ be a DFA for the {\em reversal language}
$L^\R = \{ w^\R : w \in L \}$ of the left ideal $L$
where the reversal of a word $w=a_1 a_2 \cdots a_n\in \Sigma^*$ is $w^\R = a_n \cdots a_2 a_1$.

Since $L$ is a left ideal we can assume that $F$ contains a unique final state $q_F$,
which is also a sink, i.e.\ $\delta(q_F,a) = q_F$ for all $a \in \Sigma$.

We use a simplified version of the $\mathcal{O}(\log n)$-space {\em path summary algorithm} from \cite{GHKLM18}.
A {\em path summary} is an unordered list $(Q_1, n_1) (Q_2, n_2) \dots (Q_k, n_k)$
where the $Q_i \subseteq Q$ are pairwise disjoint and nonempty
 and the $n_i \in [0,n]$ are pairwise different. 
 Note that $k \leq |Q|$ which is a constant in our setting.
 The meaning of a pair $(Q_i, n_i)$ is the following, where $w$ is the current window content: $Q_i$ is the set of all states $q$ for which
$n_i$ is the length of the shortest suffix $s$ of $w$
such that $\delta(q,s^\R) = q_F$.
Furthermore, if there exists no such suffix for a state $q$ then the state $q$ does not appear in any set $Q_i$.
Clearly $w \in L$ if and only if $q_0 \in \bigcup_{i=1}^k Q_i$.
Hence, it suffices to maintain a path summary for the current window
content $w$. We first describe, how the operations can be handled, in principle.

For $\lpop$, the algorithm removes the unique pair $(Q_i,
n_i)$ with $n=n_i$ if it exists.
For  $\rpush$,
it replaces each pair $(Q_i,n_i)$ by 
$(\{ p \in Q \setminus \{q_F\} : \delta(p,a) \in Q_i \}, n_i+1)$
if $\{ p \in Q \setminus \{q_F\} : \delta(p,a) \in Q_i \}$ is non-empty, otherwise the pair is removed from the path summary.
Finally, the pair $(\{q_F\},0)$ is added to the path summary.

It is easy to verify that the path summary is  correctly maintained,
in this way. However, some care is needed to guarantee the bounds
claimed in the statement of the theorem.

We first observe that the length $k$ of the path summary
is bounded by the constant $|Q|$. Thus the algorithm stores at most $|Q|$ many pairs $(Q_i, n_i)$.
They can be stored in a bit array of length $\mathcal{O}(\log n)$ that is divided
into $|Q|$ many chunks. Every chunk stores one pair $(Q_i, n_i)$. The state set $Q_i \subseteq Q$
is stored with $|Q|$ many bits. The number $n_i$ can be stored with $\log n$ many bits. Every chunk has an additional {\em activity bit} 
that signals whether the chunk is active or not. This is needed since 
in the path summary the algorithm  has to be able to remove and add pairs $(Q_i, n_i)$. 
If the activity bit is set to $0$ then the chunk is released and can
used for a new pair, later on. In addition the algorithm stores the window size $n$.

We have seen in the case $L\in\Len$, how the kind of counters
that are needed for a path summary and the window size can, in principle,
be implemented such that a single update works with a constant number
of operations that manipulate only $\bigO(\log \log n)$ bits. However,
there are still two challenges that need to be mastered: (1) the
counters $n_i$ need to be compared with the number $n$ (the window size), which itself
can change, and (2) when a chunk is
deactivated, its counter $n_i$ may represent any number, but when its
re-activated it should be 0, again.

Towards (1), the algorithm maintains a second
counter, $m_i$ for each chunk, which is supposed to represent
$n-n_i$. Thus, to test $n=n_i$ it suffices to check whether $m_i=0$
and we know from the $\Len$-case how this can be done.
Note that for an $\lpop$,  $m_i$ must be decremented and for a $\rpush$, $m_i$ does not change.

Towards (2), the algorithm uses a technique that could be described as
``lazy copying'' (from $0$ and from $n$, respectively). The algorithm stores two additional numbers $b_i$ and $d_i$ of size
$\bigO(\log \log n)$ for each counter $m_i$ and one additional number $a_i$ of  size
$\bigO(\log \log n)$ for each counter $n_i$.

We first describe, how to deal with $n_i$. If chunk $i$ becomes
activated, $n_i$ is supposed to be initialized to 0, but the actual
memory that it occupies might consist of arbitrary bits (inherited from the
previous counter for which the chunk was used). Overwriting these bits
would require $\Theta(\log n)$ bit operations, but the algorithm only
can manipulate  $\Theta(\log \log n)$ bits per operation. Therefore,
$n_i$ is represented by some bits of chunk $i$ and all other bits are
considered as being zero (independently of what they actually are). The additional number $a_i$ tells how many
of the last bits of the chunk for $n_i$ are valid for the representation of
$n_i$. When chunk $i$ is activated, $n_i$ should be 0 and therefore
$a_i$ can be set to 0, signifying that all bits of $n_i$ are zero. In each subsequent step, $a_i$ is
incremented by 2 and two additional positions in the chunk are set to
zero until $a_i$ equals the length of the bit string.

For $m_i$, we use a similar technique but the situation is slightly
more complicated since $m_i$ should be initially set to $n$. If the
marked position of $n$ is position $k$ then both numbers $b_i$ and
$d_i$ are set to $k$. This indicates that all positions of $m_i$ up to
$b_i$ are as in $n$ and all positions from $d_i$ on are as in
$n$. Which clearly means that $m_i$ is $n$, as required.
Subsequently, $b_i$ is decremented by 2 in each step, $d_i$ is
incremented by 2 in each step and the respective bits are copied from
$n$ to the chunk. Note that these copied bits of $n$ have not changed there
values since chunk $i$ has been activated (this holds since we copy 2 bits in each step).
\end{proof}
The lower bounds for the $\log\log n$-time level follow again from known space lower bounds
in the one-way model \cite{GHKLM18}.

\begin{theorem} \label{thm-loglog-time-lower bound}
For a regular language $L$ and sliding-window model $M$, every non-uniform
 $M$-algorithm with word size $1$ for $L$ has 
 latency at least $\log\log n - \bigO(1)$ for infinitely many $n$ in each of the following  two cases:
  \begin{enumerate}[(a)]
  \item $M  = \oV$ and $L$ is not trivial,
  \item $M  = \oF$ and $L\not\in \langle \Len, \SL\rangle$,
  \end{enumerate}
 In (a) the lower bound 
  $\log n - \bigO(1)$ holds for all $n$.
\end{theorem}

\begin{proof}
Statements (a) and (b) follow from Lemma~\ref{lemma-S-T}  using the following space lower bounds:
If $L$ is not trivial then every $\oV$-algorithm for $L$ needs space $\Omega(\log n)$ \cite[Section~6]{GHKLM18} and 
if $L$ does not belong to $\langle \Len, \SL\rangle$ then every $\oF$-algorithm for $L$ needs space $\Omega(\log n)$
for infinitely many $n$ \cite[Theorem 5.1]{GHKLM18}.
\end{proof}
\begin{corollary}
In each of the cases of Theorem~\ref{thm-loglog-time-lower bound}, any $M$-algorithm with latency $\bigO(1)$ requires word size $\Omega(\log\log n)$.
\end{corollary}

We finally turn to regular languages that have constant latency sliding window algorithms on a RAM with word size $\bigO(1)$.

\begin{theorem} \label{thm-constant-time}
Let $L$ be a regular language and $M$ a model. In the following cases, $L$ has a constant latency $M$-algorithm with word size $\bigO(1)$:
\begin{enumerate}[(a)]
\item $M = \tV$ and $L$ is trivial.
\item $M = \tF$ and $L \in \langle \Len, \Fin \rangle$
\item $M = \oF$ and $L \in \langle \Len, \SL \rangle$
\end{enumerate}
\end{theorem}

\begin{proof}
Point (a) is trivial. Point (b) holds since for every $L \in \Len$ and $L \in \Fin$ there exists a constant latency $\tF$-algorithm for $L$.
Both facts are easy: For $L \in \Len$, the algorithm only has to check in the preprocessing phase whether the fixed window size $n$
is part of a fixed semilinear set (that only depends on $L$). For $L \in \Fin$ the algorithm can always reject if the window size $n$
is larger than the maximal length of a word in $L$. Otherwise, $n$ is bounded by a fixed constant and the the algorithm can store
the window content explicitly in space and time $\bigO(1)$.
 
Point (c) follows from 
\cite[Theorem~6.1]{GHKLM18}, which says that 
$S^{\oF}_{L}(n) \in \bigO(1)$ if and only if 
 $L \in \langle \Len, \SL \rangle$. 
This implies that there is a constant time $\oF$-algorithm  since there is only a constant number of bits that have
to be manipulated.
\end{proof}

\section{Context-free languages} \label{sec-cfl}

One natural question is which extensions of the regular languages admit constant latency sliding window algorithms.
There is certainly no hope to go up to the whole class of context-free languages as this would yield a linear
time algorithm for parsing context-free languages.
We will show that, even for real-time deterministic context-free languages, there is also
little hope to find a constant latency uniform fixed-size sliding window
algorithm even though all deterministic context-free languages can be
parsed in linear time.
More precisely, we will construct a real-time deterministic context-free language
for which we prove a lower bound of $\Omega(n^{1/2-o(1)})$ per update, under the OMV conjecture~\cite{HenzingerKNS15}. 

On the positive side, in this section we will present constant latency sliding window algorithms for the class of visibly pushdown languages
and algorithms with logarithmic latency for deterministic 1-counter languages.

\subsection{Lower bound for real-time deterministic context-free languages}\label{sec:dcfl}

First let us recall the online matrix-vector multiplication problem
that we will use in our reduction.  The \emph{Online Matrix-Vector
multiplication} (OMV) problem is the following: the input consists of a
Boolean matrix $M \in \{0,1\}^{n \times n}$ and $n$ Boolean vectors
$V_1, \dots, V_n \in \{0,1\}^{n \times 1}$ and the vector $M\cdot V_i$
must be computed before the vector $V_{i+1}$ is read.  The OMV conjecture \cite{HenzingerKNS15} states that a RAM with register length $\log n$ cannot solve the OMV
problem in time $\mathcal{O}(n^{3-\epsilon})$ for any $\epsilon>0$.
The OMV conjecture implies tight lower bounds for a number of important problems, like
subgraph connectivity, Pagh’s problem, $d$-failure connectivity, decremental single-source shortest paths, and decremental transitive closure;
see \cite{HenzingerKNS15}.

Recall that a real-time deterministic context-free language $L$ is a language that is accepted by
a deterministic pushdown automaton without $\varepsilon$-transitions. Hence, the automaton reads
an input symbol in each computation step.

\begin{figure}
\centering
\mbox{
$\begin{bmatrix}
m_{11} & m_{12} & m_{13} \\ m_{21} & m_{22} & m_{23} \\ m_{31} & m_{32} & m_{33}
\end{bmatrix}
\begin{bmatrix}
v_1 \\ v_2 \\ v_3
\end{bmatrix} \; \longrightarrow \;
\$ \, m_{13} \, m_{12} \, m_{11} \, \$ \, m_{23} \, m_{22} \, m_{21} \, \$ \, m_{33} \, m_{32} \, m_{31} \, \# \, v_1 \, v_2 \, v_3$}
\caption{Encoding of a matrix-vector product.}
\label{fig:omv}
\end{figure}

\begin{lemma} \label{lemma-omv}
There exists a real-time deterministic context-free language $L$ such that
any (uniform) $\oF$-algorithm for $L$ with
logarithmic word size and latency $t(n)$ yields an
algorithm for the OMV problem with latency $\bigO(n^2) \cdot t(\bigO(n^2))$.
\end{lemma}

\begin{proof}
Let $n$ be the dimension of the OMV problem that we want to solve. We define the window size
$m = 3n+n^2+1 = \Theta(n^2)$. The reduction will be based on a language $L \subseteq \{a,0,1,\$, \#\}^*$ that contains the word 
$a^j \text{enc}(M) \# \text{enc}(V) a^{n-j}$
(for $M \in \{0,1\}^{n \times n}$, $V \in \{0,1\}^{n \times 1}$, and $1 \le j \le n$) if and only if $M \cdot V$ 
contains a $1$ on its $j$-th coordinate, for an encoding function
$\text{enc}$ that we now present. Since we do not care about
strings that are not of this form, it does not matter which of them
are in $L$. In fact, a string of the wrong form shall be in $L$, if
and only if the automaton below accepts it.

The matrix/vector encoding is illustrated in Figure~\ref{fig:omv}.
A Boolean vector $V = [v_1, \dots, v_n]^{\mathsf{T}} \in \{0,1\}^{n \times 1}$
is encoded as the binary string $v_1 \cdots v_n \in \{0,1\}^n$.
A matrix $M \in \{0,1\}^{n \times n}$ is encoded row
by row, with the first row first, and the last row last. Each 
row starts with the dedicated symbol \$, followed by the encoding of
the row (a word over the alphabet $\{0,1\}$) in reverse order. Thus, the encoding of the $j$-th
row starts with $M_{j,n}$ and ends with $M_{j,1}$).

We can construct a real-time deterministic pushdown automaton $\mathcal{P}$
which accepts the word $a^j \text{enc}(M) \# \text{enc}(V) a^{n-j}$ for $M,V,j$ as above,
if and only if the $j$-th entry of the product $M \cdot V$ is $1$:
The pushdown automaton reads the $a^j$-prefix on the stack and
skips to the $j$-th row encoding of the matrix by popping $a$ from the stack on every read row delimiter $\$$.
Then it reads the $j$-th row of $M$ in reverse on the stack, skips to the encoding of $V$,
and can verify whether the $j$-th row of $M$ and $V$ both have $1$-bits at a common position.

Now let us suppose that we have a uniform fixed-size sliding window
algorithm $\mathcal{A}$  for $L$ with
logarithmic word size and latency $t(n)$.
Given the matrix $M$, we initialize an instance of $\mathcal{A}$ with window size $m = 3n+n^2+1 = \Theta(n^2)$
and fill the sliding window with $a^{2n} \text{enc}(M)\#$.
This word has length $\mathcal{O}(n^2)$, so this initial preprocessing takes time $\mathcal{O}(n^2) \cdot t(\mathcal{O}(n^2))$.

From the data structure prepared with the window content $a^{2n} \text{enc}(M)\#$
we can get the last bit of the vector $M \cdot V_1$ by loading $V_1$ into the window with $n$ updates
each taking time $t(m) = t(\mathcal{O}(n^2))$. This yields the window content $a^{n} \text{enc}(M)\# \text{enc}(V_1)$,
which is  in $L$ if and only if
$(M \cdot V_1)_n=1$. Then loading an $a$ into the window 
we obtain the window content $a^{n-1} \text{enc}(M)\# \text{enc}(V_1) a$. It belongs to $L$ if and only if $(M \cdot V_1)_{n-1} = 1$. We
repeat the process to obtain all bits of $M \cdot V_1$. Computing $M \cdot V_1$ thus
requires $2n-1$ updates and thus time $\mathcal{O}(n) \cdot
t(\mathcal{O}(n^2))$.  After computing $M \cdot V_1$ we need to
compute $M \cdot V_2$, and the other products $M \cdot V_i$ next, and
for that we would like to reset the algorithm
$\mathcal{A}$ so that the data structure is prepared with the word
$a^{2n} \text{enc}(M) \#$, again. Here the final ``trick'' is applied:
instead of doing a new initialization for each vector, requiring
$\mathcal{O}(n^2) \cdot t(\mathcal{O}(n^2))$ steps each time, we
rather do a rollback: to this end, the sliding window algorithm
$\mathcal{A}$ is modified so that each time it
changes the value of some register $\ell$ in memory, it writes $\ell$
and the old value of register $\ell$ into a log. By
rolling back this log it is  able to undo all changes during the
processing of $V_1$. The extra running time for keeping the log and rolling back the computation is proportional to
the number of changes in memory and thus needs time only $\mathcal{O}(n)
\cdot t(\mathcal{O}(n^2))$. 

Overall the time to deal with one vector is $\mathcal{O}(n) \cdot t(\mathcal{O}(n^2))$
which yields a total running time of $\mathcal{O}(n^2) \cdot t(\mathcal{O}(n^2))$ for OMV.
\end{proof}
Lemma~\ref{lemma-omv} states that a sliding window for $L$ with
logarithmic word size and latency
$t(n)=\mathcal{O}(n^{1/2-\epsilon})$ would yield an algorithm for
OMV with a running time of $\mathcal{O}(n^{3-2\epsilon})$, which would
contradict the OMV conjecture.

\begin{corollary}
There exists a fixed  deterministic context-free language $L$ such that,
conditionally to the OMV conjecture, there is no (uniform) $\oF$-algorithm for $L$ with logarithmic word
size and latency $n^{1/2-\epsilon}$ for any $\epsilon>0$.
\end{corollary}

\subsection{Visibly pushdown languages} \label{sec:vpl}

In this section, we provide a constant latency $\tV$-algorithm for visibly pushdown languages.
 We first define visibly pushdown automata and their languages (for more details see  \cite{AlurM04}) and
then show the upper bound result.

Visibly pushdown automata are like general pushdown automata, but the input alphabet is partitioned into 
call letters (that necessarily trigger a push operation on the stack), return letters 
(that necessarily trigger a pop operation on the stack), and internal letters (that do not change the stack).
Formally, a {\em pushdown alphabet} is a triple $\Sigma = (\Sigma_c,\Sigma_r,\Sigma_{\mathit{int}})$
consisting of three pairwise disjoint alphabets:
a set of {\em call letters} $\Sigma_c$, a set of {\em return letters} $\Sigma_r$
and a set of {\em internal letters} $\Sigma_{\mathit{int}}$.
We identify $\Sigma$ with the union $\Sigma = \Sigma_c \cup \Sigma_r \cup \Sigma_{\mathit{int}}$.
The set $W$ of {\em well-nested} words  over $\Sigma$ is defined as the smallest set 
such that (i) $\{\eps\} \cup \Sigma_{\mathit{int}} \subseteq W$, (ii) $W$ is closed under
concatenation and (iii) if  $w \in W$, $a \in \Sigma_c$, $b \in \Sigma_r$ then also $awb \in W$.  
Every well-nested word over $\Sigma$ can be uniquely written as a product of
{\em Dyck primes} $D = \Sigma_{\mathit{int}} \cup \{ a w b :  w \in W, a \in \Sigma_c, b \in \Sigma_r \}$
($W$ is a free submonoid of $\Sigma^*$ that is freely generated by $D$).
Note that every word $w \in \Sigma^*$ has a unique factorization $w = s t u$ with 
$s \in (W \Sigma_r)^*$, $t \in W$ and $u \in (\Sigma_c W)^*$. To see this, note that the maximal well-matched
factors in a word $w \in \Sigma^*$ do not overlap. If these maximal well-matched
factors are removed from $w$ then a word from $\Sigma_r^* \Sigma_c^*$ must remain (other one of the 
removed well-matched factors would be not maximal); see also \cite[Section~5]{Ganardi19}.

A {\em visibly pushdown automaton (VPA)} is a tuple $\mathcal{A} = (Q,\Sigma,\Gamma,\bot,q_0,\delta, F)$
where $Q$ is a finite state set, $\Sigma$ is a pushdown alphabet, $\Gamma$ is the finite stack alphabet
containing a special symbol $\bot$ (representing the bottom of the stack), $q_0 \in Q$ is the initial state, $F \subseteq Q$ is the set of final states and $\delta = \delta_c \cup \delta_r \cup \delta_{\mathit{int}}$ is the transition function
where $\delta_c \colon Q \times \Sigma_c \to (\Gamma \setminus \{\bot\}) \times Q$,
$\delta_r \colon Q \times \Sigma_r \times \Gamma \to Q$
and $\delta_{\mathit{int}} \colon Q \times \Sigma_{\mathit{int}} \to Q$.
The set of {\em configurations} $\Conf$ is the set of all words $\alpha q$
where $q \in Q$ is a state and $\alpha \in  \bot (\Gamma \setminus \{\bot\})^*$
is the {\em stack content}. 
We define $\hat{\delta} \colon \Conf \times \Sigma \to \Conf$ as follows, where $p \in Q$ and $\alpha \in \bot (\Gamma \setminus \{\bot\})^*$:
\begin{itemize}
	\item If $a \in \Sigma_c$ and $\delta(p,a) = (\gamma,q)$
	then $\hat{\delta}(\alpha p,a) = \alpha \gamma q$.
	\item If $a \in \Sigma_{\mathit{int}}$ and $\delta(p,a) = q$
	then $\hat{\delta}(\alpha p, a) = \alpha q$.
	\item If $a \in \Sigma_r$, $\gamma \in \Gamma \setminus \{\bot\}$, and $\delta(p,a,\gamma) = q$ 
	then $\hat{\delta}(\alpha \gamma p, a) = \alpha q$.
	\item If $a \in \Sigma_r$ and $\delta(p,a,\bot) = q$
	then $\hat{\delta}(\bot p,a) = \bot q$ (so $\bot$ is not popped from the stack).
\end{itemize}
As usual we inductively extend $\hat{\delta}$ to a function $\hat{\delta}^* \colon \Conf \times \Sigma^* \to \Conf$
where $\hat{\delta}^*(c,\eps) = c$ and $\hat{\delta}^*(c,wa) = \hat{\delta}(\hat{\delta}^*(c,w),a)$ for all $c \in \Conf$, $w \in \Sigma^*$ and $a \in \Sigma$.
In the following, we write $\hat{\delta}$ for $\hat{\delta}^*$.
The {\em initial} configuration is $\bot q_0$ and a configuration $c$ is {\em final} if $c \in \Gamma^* F$.
A word $w \in \Sigma^*$ is {\em accepted} from a configuration $c$
if $\hat{\delta}(c,w)$ is final.
The VPA $\mathcal{A}$ {\em accepts} $w$ if $w$ is accepted from the initial configuration.
The set of all words accepted by $\mathcal{A}$ is denoted by $L(\mathcal{A})$;
the set of all words accepted from $c$ is denoted by $L(c)$.
A language $L$ is a {\em visibly pushdown language (VPL)} if $L = L(\mathcal{A})$ for some VPA $\mathcal{A}$.

One can also define nondeterministic visibly pushdown automata in the usual way; they can always be converted into deterministic ones \cite{AlurM04}.
This leads to good closure properties of the class of all VPLs like closure under Boolean operations, concatenation and Kleene star.

As usual, we denote with $Q^Q$ the monoid
of all mappings from the state set $Q$ to $Q$ with composition of functions as the monoid operation (see also Section~\ref{sec-daba}).
Notice that $\mathcal{A}$ can only see the top of the stack when reading return symbols.
Therefore, the behavior of $\mathcal{A}$ on a well-nested word
is determined only by the current state and independent of the current stack content.
We can therefore define a mapping $\phi \colon W \to Q^Q$ by 
$\hat{\delta}(\bot p, w) = \bot \, \phi(w)(p)$ for all $w \in W$ and $p \in Q$.
Note that this implies $\hat{\delta}(\alpha p, w) = \alpha \, \phi(w)(p)$ for all $w \in W$ and $\alpha p \in \Conf$.
The mapping $\phi$ is a monoid morphism (recall that $W$ is a monoid with respect to concatenation).

For the further consideration, it is useful to extend $\phi$ to a mapping $\phi \colon \Sigma^* \to Q^Q$, which will
be no longer a monoid morphism. To do this we first define 
$\phi(a) : Q \to Q$ for letters $a \in \Sigma_r \cup \Sigma_c$ by
$\delta(p,a,\bot) = \phi(a)(p)$ for $a \in \Sigma_r$ and $\delta(p,a) = (\gamma, \phi(a)(p))$ for $a \in \Sigma_c$ (and some $\gamma \in \Gamma \setminus \{\bot\}$).
Note that for a return letter $a$, $\phi(a)$ is the state transformation induced by the letter $a$ on the stack only containing $\bot$.
Consider now an arbitrary word $w \in \Sigma^*$. As mentioned above, there is a unique factorization 
\begin{equation} \label{canonical-fact}
w = w_0 a_1 w_1 a_2 w_2 \cdots a_k w_k \in  (W \Sigma_r)^* W (\Sigma_c W)^*
\end{equation}
such that $k \ge 0$, $w_i \in W$ for all $0 \le i \le k$ and for some $s \in [0,k]$ (the {\em separation position}) we 
have $a_1, \ldots, a_s \in \Sigma_r$ and $a_{s+1}, \ldots, a_k \in \Sigma_c$.
We then define the mapping $\phi(w) : Q \to Q$ as
$\phi(w)   = \phi(w_0) \phi(a_1) \phi(w_1) \cdots \phi(a_k) \phi(w_k)$.
Again, notice that the mapping $\phi : \Sigma^* \to Q^Q$ is not a monoid homomorphism. However, 
$\phi$ captures the behaviour of $\mathcal{A}$ on a word $w$:

\begin{lemma}\label{claim-delta}
  For every word $w \in \Sigma^*$ and all states $q \in Q$ we have
  $\hat{\delta}(\bot q, w) = \alpha \phi(w)(q)$ for some stack content
  $\alpha \in \bot (\Gamma \setminus \{\bot\})^*$.
\end{lemma}

\begin{proof}
Let $w$ be decomposed as in \eqref{canonical-fact}.
We prove the lemma by induction on $k$. If $k=0$ then $w = w_0$ is well-nested.  By definition of the monoid morphism $\phi \colon W \to Q^Q$
we have
\[
\hat{\delta}(\bot q, w) = \bot \, \phi(w)(q).
\]
Now assume that $k > 0$ and let $w' = w_0 a_1 w_1 a_2 w_2 \cdots a_{k-1} w_{k-1}$ so that $w = w' a_{k} w_{k}$. 
By induction, there is a stack content $\alpha' \in \bot (\Gamma \setminus \{\bot\})^*$
such that $\hat{\delta}(\bot q, w') = \alpha' \phi(w')(q)$. Let $q' =  \phi(w')(q)$ so that $\hat{\delta}(\bot q, w') = \alpha' q'$.
We distinguish two cases:

If $a_k \in \Sigma_r$, then $a_1, a_2, \ldots, a_k \in \Sigma_r$. We then must have $\alpha' = \bot$, i.e.,
$\hat{\delta}(\bot q, w') = \bot q'$, and obtain
\begin{align*}
\hat{\delta}(\bot q, w' a_k) = & \ \hat{\delta}(\hat{\delta}(\bot q, w'), a_k) = \hat{\delta}(\bot q', a_k) = \bot \delta(q',a_k,\bot) \\
 = &\  \bot \phi(a_k)(q') =  \bot  (\phi(w') \phi(a_k)) (q) =  \bot  \phi(w' a_k)(q) .
\end{align*}
Finally, we get
\begin{align*}
\hat{\delta}(\bot q, w' a_k w_k) = & \ \hat{\delta}(\hat{\delta}(\bot q, w' a_k), w_k) = \hat{\delta}(\bot \phi(w' a_k)(q), w_k) \\
 = & \ \bot \phi(w_k)( \phi(w' a_k)(q) ) = \bot \phi(w' a_k w_k)(q).
\end{align*}
Now assume that $a_k \in \Sigma_c$. Note that there is $\gamma \in \Gamma \setminus \{\bot\}$ such that 
$\delta(q', a_k) = (\gamma, \phi(a_k)(q')) = (\gamma, \phi(w' a_k)(q))$.
We obtain 
\begin{align*}
\hat{\delta}(\bot q, w' a_k) = & \ \hat{\delta}(\hat{\delta}(\bot q, w'), a_k) = \hat{\delta}(\alpha' q', a_k) = 
\alpha' \gamma \phi(w' a_k)(q).
\end{align*}
Finally, as in case $a_k \in \Sigma_r$ we obtain $\hat{\delta}(\bot q,
w' a_k w_k) = \alpha' \gamma \phi(w' a_k w_k)(q)$.
This concludes the proof of the lemma.
\end{proof}
%To exclude some pathological cases we assume that $\Sigma_c \neq \emptyset$ and $\Sigma_r \neq \emptyset$.
%In fact, if $\Sigma_c = \emptyset$ or $\Sigma_r = \emptyset$ then any VPL over that pushdown alphabet would be regular.
Now we are ready to state the main result of this section.

\begin{theorem}\label{thm-vpl}
Every visibly pushdown language $L$ has a $\tV$-algorithm with
latency $\bigO(1)$, space complexity $\bigO(n \log n)$ and word size $\bigO(\log n)$.
\end{theorem}
For the proof of Theorem~\ref{thm-vpl} we will use the $\tV$-algorithm for 
$\mathsf{prod}_{\mathcal{M}}$ from the proof of Theorem~\ref{thm-DABA-new} ($\mathcal{M}$ is again a finite monoid). 
In the following, we call the data structure behind this algorithm a {\sf DABA}$(\mathcal{M})$ data structure, where {\sf DABA} stands for deamortized banker's algorithm (the name
for the data structure in  \cite{TangwongsanH017} used for the proof of Theorem~\ref{thm-cell-probe-old}).  
In the proof of Theorem~\ref{thm-DABA-new}, {\sf DABA}$(\mathcal{M})$ stores a sequence of monoid elements $m_1, m_2, \ldots, m_k \in \mathcal{M}$
and a $\mathsf{prod}_{\mathcal{M}}$-query returns the monoid product $m_1 m_2 \cdots m_k$.
For our application of the {\sf DABA}$(\mathcal{M})$ data structure to visibly pushdown languages in the next section,
we have to store sequences of pointers $p_1, p_2, \ldots, p_k$. Following pointer $p_k$ we can determine a monoid element $m_i$ (and some additional data values that will
be specified below).
When applying a $\mathsf{prod}_{\mathcal{M}}$-query to such a sequence of pointers, the monoid product $m_1 m_2 \cdots
m_k \in \mathcal{M}$ is returned. In addition we have the update-operations from $\tV$-model that allow to remove 
the first or last pointer or to add a new pointer at the beginning or end. All operations work in constant time.

\begin{proof}[Proof of Theorem~\ref{thm-vpl}]
Let us fix a (deterministic) VPA $\mathcal{A} = (Q,\Sigma,\Gamma,\bot,q_0,\delta, F)$ with 
$\Sigma = (\Sigma_c,\Sigma_r,\Sigma_{\mathit{int}})$. 
By Lemma~\ref{claim-delta}, a variable-size sliding window
algorithm for $L(\mathcal{A})$ only needs  to maintain the state transformation $\phi(w)$ for the current window content $w$.

\begin{figure}
\begin{center}
\begin{tikzpicture}[yscale=-0.8]

\tikzstyle{blob}=[draw=none,fill=yellow, inner sep = 1pt]

\node[blob, label={above:$T(w)$}] at (0,0.6) (r) {$\phi(w)$};

\node[draw, minimum size = .5cm] (a) at (-0.75,1.2) {};
\node[circle, fill, inner sep = 1.5pt] at (-0.75,1.2) {};
\node[draw, minimum size = .5cm] (b) at (-0.25,1.2) {};
\node[circle, fill, inner sep = 1.5pt] at (-0.25,1.2) {};
\node[draw, minimum size = .5cm] (c) at (0.25,1.2) {};
\node[circle, fill, inner sep = 1.5pt] at (0.25,1.2) {};
\node[draw, minimum size = .5cm] (d) at (0.75,1.2) {};
\node[circle, fill, inner sep = 1.5pt] at (0.75,1.2) {};

\node[blob] (a') at (-4,3) {$(a_1, \phi(u_1), b_1)$};
\node[blob] (b') at (-1,3) {$(a_2, \phi(u_2), b_2)$};
\node[minimum height = 1cm] (c') at (2,3) {$\dots$};
\node[minimum height = 1cm] (d') at (5,3) {$\dots$};

\draw (a.center) -- (a');
\draw (b.center) -- (b');
\draw (c.center) -- (c');
\draw (d.center) -- (d');

\node[fill = white, above = 0.5em of a', inner sep = 0] {$T(w_1)$};
\node[fill = white, above = 0.5em of b', inner sep = 0] {$T(w_2)$};

\node[blob] (a'') at (-4,5) {$(a,\phi(a))$};
\node[blob] (b'') at (-1,5) {$\phi(u_2)$};

\draw (a') -- (a'');
\draw (b') -- (b'');

\node[fill = white, above = 0.5em of a'', inner sep = 0] {$T(u_1)$};
\node[fill = white, above = 0.5em of b'', inner sep = 0] {$T(u_2)$};

\node[draw, minimum size = .5cm] (b1) at (-2,5.6) {};
\node[circle, fill, inner sep = 1.5pt] at (-2,5.6) {};
\node[draw, minimum size = .5cm] (b2) at (-1.5,5.6) {};
\node[circle, fill, inner sep = 1.5pt] at (-1.5,5.6) {};
\node[draw, minimum size = .5cm] (b3) at (-1,5.6) {};
\node[circle, fill, inner sep = 1.5pt] at (-1,5.6) {};
\node[draw, minimum size = .5cm] (b4) at (-0.5,5.6) {};
\node[circle, fill, inner sep = 1.5pt] at (-0.5,5.6) {};
\node[draw, minimum size = .5cm] (b5) at (0,5.6) {};
\node[circle, fill, inner sep = 1.5pt] at (0,5.6) {};

\draw (b1.center) -- ++ (-0.5,0.5);
\draw (b2.center) -- ++ (-0.2,0.5);
\draw (b3.center) -- ++ (0,0.5);
\draw (b4.center) -- ++ (0.2,0.5);
\draw (b5.center) -- ++ (0.5,0.5);

\end{tikzpicture}
\end{center}
\caption{The data structure of the $\tV$-algorithm for VPLs.
A well-nested word $w$ decomposes into Dyck primes $w = w_1 w_2 \cdots w_k$ where each $w_i$ is either an internal letter
or consists of a call letter $a_i$, a well-nested word $u_i$ and a return letter $b_i$.
The node of $w$ stores the state transformation $\phi(w)$ together with a DABA$(Q^Q)$ instance, which maintains a list of the children $w_i$ and their
state transformations  $\phi(w_i)$.}
\label{fig:vpl}
\end{figure}
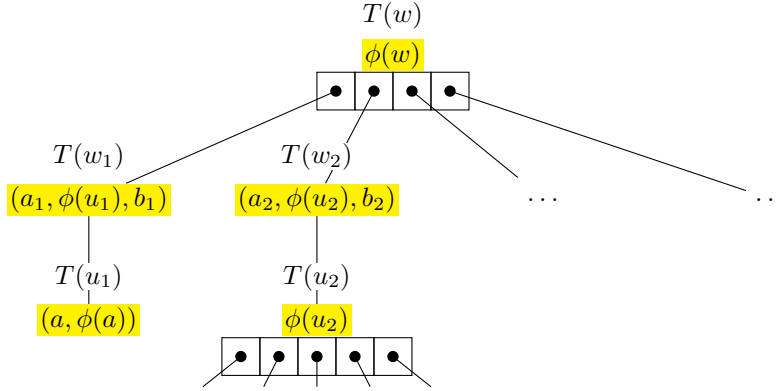

With a well-nested word $w \in W$ we associate a node-labelled ordered tree $T(w)$ as follows, where $\text{id}_Q \in Q^Q$ is the 
identity mapping on $Q$:
\begin{enumerate}
\item If $w  = \eps$ then $T(w)$ consists of a single node labelled with the pair $(\eps, \text{id}_Q)$.
\item If $w = a \in \Sigma_{\mathit{int}}$ then $T(w)$ consists of a single node labelled with $(a,\phi(a))$.
\item If $w$ is a Dyck prime $a u b$ with $a \in \Sigma_c$, $u \in W$ and $b \in \Sigma_r$,
then $T(w)$ consists of a root node,
labelled with $(a, \phi(w), b)$, with a single child which is the root of the tree $T(u)$.
\item If $w = w_1 \cdots w_k$ for Dyck primes $w_1, \ldots, w_k$ and
  $k \geq 2$, the root of $T(w)$ is labelled
with $\phi(w)$ and it has  the roots of the trees $T(w_1), \ldots,
T(w_k)$ $k$ as children (from left to right).
\end{enumerate}
See Figure~\ref{fig:vpl} for an illustration of the tree $T(w)$.
We also speak of a tree $T(w)$ of type ($i$) (for $1 \leq i \leq 4$). Moreover, we say that a node $v$ is of type ($i$) if the subtree
rooted in $v$ is of type ($i$). Note that the type of node can be obtained from its label.

We have to implement several operations on such trees. In order to spend only constant time for each of the operations,
we maintain the children $v_1, \ldots, v_k$ of a node $v$ of type (4)
as a DABA$(Q^Q)$ data structure which we denote by $\mathsf{DABA}(v)$ 
(note that $v_i$ is the root of the tree $T(w_i)$). This data structure is needed in order to maintain
the value $\phi(w) = \phi(w_1) \cdots \phi(w_k)$ (the label of $v$). For the implementation of $\mathsf{DABA}(v)$ we have
to slightly extend the data structure from the proof of Theorem~\ref{thm-DABA-new}: there, every entry of the data structure stores
an element of the monoid $\mathcal{M}$ (here, $Q^Q$). Here, the entries of $\mathsf{DABA}(v)$ are pointers to the children 
$v_1, \ldots, v_k$. Note that from $v_i$ we can obtain in constant time the monoid value $m_i := \phi(w_i) \in Q^Q$ as its label.
The partial products $m_i \cdots m_j$ for the guardians as well as the additional $\mathcal{O}(1)$ many pointers to entries of the DABA data
structure are treated as in the proof of Theorem~\ref{thm-DABA-new}.
For the purpose of maintaining $\phi(w)$ only the monoid elements of $m_1, \ldots, m_k$ are relevant, but 
we use $\mathsf{DABA}(v)$ also in order to store the list of children of $v$ and hence the tree structure of $T(w)$.
Intuitively, a tree $T(w)$ can be seen as a nested DABA data structure for the well-nested word $w$.

Assume now that the current  window content is $w \in \Sigma^*$ and consider the unique factorization for 
$w$ in \eqref{canonical-fact} with the separation position $s \in [0,k]$. For a symbol $a \in \Sigma_c \cup \Sigma_r$ 
we write $\langle a \rangle$ for the pair $(\phi(a), a)$ below.
Our $\tV$-algorithm stores on the top level two lists and a tree (here, again, every
tree $T(w_i)$ is represented by a pointer to its root):
\begin{itemize}
\item the descending list $L_{\downarrow}  =[T(w_0), \langle a_1 \rangle, T(w_1), \langle a_{2} \rangle, \ldots, T(w_{s-1}), \langle a_s \rangle]$
\item the separating tree $T_s = T(w_s)$
\item the ascending list $L_{\uparrow} = [ \langle a_{s+1} \rangle, T(w_{s+1}), \ldots, \langle a_k \rangle, T(w_k)]$.
\end{itemize}
These two lists are maintained by the DABA data structures
$\mathsf{DABA}_{\downarrow}$ and $\mathsf{DABA}_{\uparrow}$,
respectively. 
These DABA data structures maintain the aggregated state transformations $\phi(w_0 a_1 w_1 a_2 \cdots w_{s-1} a_s)$ and
$\phi(a_{s+1} w_{s+1} \cdots a_k w_k)$ from which, together with $\phi(w_s)$ (which can be obtained from the root of $T_s$)
 we can obtain the value $\phi(w)$, which, in turn, allows to check whether $w \in L(\mathcal{A})$ by Claim~\ref{claim-delta}.
Note that $L_{\downarrow} = []$ (the empty list) in case $s = 0$ and $L_{\uparrow} = []$ in case $k = s$.
Since each of lists $L_{\downarrow}$ and $L_{\uparrow}$ can be empty, 
we will need all four types of operations ({\sf leftpop}, {\sf rightpop}, {\sf leftpush}, and {\sf rightpush}) for 
$L_{\downarrow}$ and $L_{\uparrow}$. Therefore, the symmetric DABA data structure 
from Theorem~\ref{thm-DABA-new} is really needed here (even if we only would aim for a $\oV$-algorithm).

Building on the above data structure we now show how to implement
all window update operations in constant time on a RAM with logarithmic word size. Before we come to the window updates, we 
first introduce a few auxiliary operations.

\medskip
\noindent
{\sf concatenate}$(T(u), T(v))$: Here, $T(u)$ is a tree of any type ($i$) and $v \in D$ (i.e., $T(v)$ has type (2) or (3)). The operation returns
the tree $T(uv)$. Let $r_u$ ($r_v$) be the root of $T(u)$ ($T(v)$).
The case that $u = \eps$ is trivial (just return $T(v)$). Otherwise $T(uv)$ is a tree of type (4). 
If $u \in D$, then we have to create a new root node $r$ with the two children $r_u$ and $r_v$. Moreover, we have
to initialize the structure $\mathsf{DABA}(r)$ containing $r_u$ and $r_v$.
If $T(u)$ is of type (4), then $r_v$ becomes the right-most child of $r_u$. For this we make a $\mathsf{rightpush}$-operation on
$\mathsf{DABA}(r_u)$.  Then $T(u)$ is returned.

\medskip
\noindent
For the following operations, let $r$ be the root node of the tree $T(w)$.

\medskip
\noindent
{\sf left-prime-pop}$(T(w))$:  It only applies to trees $T(w)$ with $w \in W \setminus \{\eps\}$.
Hence, we can write $w = uv$ with $u \in D$ and $v \in W$. The operation returns the pair $[T(u),T(v)]$. %Let $r$ be the root of $T(w)$.
The case that $v = \eps$ (i.e., $T(w)$ has type (2) or (3)) is obvious. Hence, let us assume that $v \in W \setminus \{\eps\}$ in which case
$T(w)$ is of type (4). Then $T(u)$ is the tree rooted in the left-most child of $T(w)$, which can be accessed in constant
time. In order to get the tree $T(v)$ we make a $\mathsf{leftpop}$-operation on $\mathsf{DABA}(r)$. Let $T'$ the resulting
tree and $r'$ its root. If $r'$ has at least two children then $T(v) = T'$. Otherwise $T(v)$ is the tree rooted in the unique child of $r'$.
In both cases, we obtain $T(v)$ in constant time.

\medskip
\noindent
{\sf left-symbol-pop}$(T(w))$: It only applies to trees $T(w)$ with $w \in W \setminus \{\eps\}$. In this case we must have
$w = a x$ for some $a \in \Sigma_{\mathit{int}} \cup \Sigma_c$. Let $[T(u),T(v)] =$ {\sf left-prime-pop}$(T(w))$.
If $T(u)$ is of type (2) then we have $u = a \in \Sigma_{\mathit{int}}$, $x = v$, and {\sf left-symbol-pop}$(T(w))$ returns $T(v)$.
If $T(u)$ is of type (3) then we have $u = a u' b$ for $a \in \Sigma_c$, $b \in \Sigma_r$, $u' \in W$, and $x = u'bv$. 
Then  {\sf left-symbol-pop}$(T(w))$ returns the  list $[T(u'), b, T(v)]$. The symbol $b$ can be obtained from the root label of 
the tree $T(u)$ and the tree $T(u')$ is the tree rooted in the unique child of $T(u)$.

\medskip
\noindent
{\sf construct-prime}$(a, T(w), b)$: We have $a \in \Sigma_c$ and $b \in \Sigma_r$ and the operation returns the tree $T(awb)$ of type (3). We simply add a new root node $r'$
to $T(w)$ whose single child is $r$. The label of $r'$ must be $(a,\phi(awb),b)$.
 The state transformation $\phi(awb)$ can be computed as follows: let $p \in Q$ and assume that 
$\delta(p,a) = (\gamma, q)$. Let $q' = \phi(w)(q)$. Then $\phi(a w b)(p) = \delta(q', b, \gamma)$.

\medskip
\noindent
Using the above operations on trees of well-nested words we can now implement the   
sliding-window update operations with constant latency. Recall that the descending list
$L_{\downarrow}$, the separating tree $T_s$, and the ascending list $L_{\uparrow}$ from the main part of the paper.

\medskip
\noindent
{\sf leftpop}: We distinguish several cases:
\begin{itemize}
\item $L_{\downarrow} = L_{\uparrow} = []$ and $T_s = T(\eps)$. In this case the window is empty and we do nothing.
\item $L_{\downarrow} = []$ and $T_s = T(w_0)$ with $w_0 \neq \eps$. 
We first call \mbox{\sf left-symbol-pop}$(T_s)$. If this returns a single tree $T(v)$, then this tree becomes $T_s$.
If \mbox{\sf left-symbol-pop}$(T_s)$ returns the list 
$[T(u'), b, T(v)]$ (with $b \in \Sigma_r$) then we update $L_{\downarrow}$ to $[T(u'), \langle b \rangle]$ and set $T_s$ to $T(v)$.
\item $L_{\downarrow} = []$, $T_s = T(\eps)$, and $L_{\uparrow} \neq []$. Then we can write 
$L_{\uparrow} = [\langle a_1 \rangle, T(w_1), \ldots]$ and the left most letter in the window is $a_1$. We therefore remove the entries $\langle a_1 \rangle$ and $T(w_1)$
from the list $L_{\uparrow}$ by doing two $\mathsf{leftpop}$-operations
on  $\mathsf{DABA}_{\uparrow}$. Finally, we set $T_s$ to $T(w_1)$.
\item $L_{\downarrow} = [T(\eps), \langle a_1 \rangle, \ldots]$. We remove $T(\eps)$ and $\langle a_1 \rangle$
from $L_{\downarrow}$ using two $\mathsf{leftpop}$-operations on $\mathsf{DABA}_{\downarrow}$.
\item $L_{\downarrow} = [T(w_0), \langle a_1 \rangle, \ldots]$ with $w_0 \neq \eps$.
We first extract the tree $T(w_0)$ from $L_{\downarrow}$ using a $\mathsf{leftpop}$-operation on $\mathsf{DABA}_{\downarrow}$.
Then we call {\sf left-symbol-pop}$(T(w_0))$. If this returns a single tree $T(v)$, we add this tree back to $L_{\downarrow}$
using a {\sf leftpush}-operation on $\mathsf{DABA}_{\downarrow}$. If {\sf left-symbol-pop}$(T(w_0))$ returns the list 
$[T(u'), b, T(v)]$ (with $b \in \Sigma_r$) then we add these entries to 
$L_{\downarrow}$ using three {\sf leftpush}-operations on $\mathsf{DABA}_{\downarrow}$.
\end{itemize}
{\sf rightpush}$(b)$: We distinguish the following cases:
\begin{itemize}
\item $b \in \Sigma_c$: Using two $\mathsf{rightpush}$-operations on $\mathsf{DABA}_{\uparrow}$ we add $\langle b \rangle$ and $T(\eps)$ to 
the list $L_{\uparrow}$.
\item $b \in \Sigma_{\mathit{int}}$ and $L_\uparrow = []$: 
We construct the tree $T(b)$, and set the separating tree $T_s$ to the result of 
{\sf concatenate}$(T_s, T(b))$.
\item $b \in \Sigma_{\mathit{int}}$ and $L_\uparrow \neq []$: Then we can write $L_\uparrow = [\ldots, \langle a_k \rangle, T(w_k)]$.
We first extract $T(w_k)$ from $L_{\uparrow}$ using a {\sf rightpop}-operation on $\mathsf{DABA}_{\uparrow}$.
Then we construct the tree $T(b)$, call {\sf concatenate}$(T(w_k), T(b))$ and add the resulting tree back to $L_{\uparrow}$ using a {\sf rightpush}-operation on $\mathsf{DABA}_{\uparrow}$.
\item $b \in \Sigma_r$ and $L_{\uparrow} = []$: We add $T_s$ and $\langle b \rangle$ to $L_{\downarrow}$ using two
{\sf rightpush}-operations on $\mathsf{DABA}_{\downarrow}$ and then set  $T_s$ to $T(\eps)$.
\item $b \in \Sigma_r$ and $L_{\uparrow} = [\langle a_k \rangle, T(w_k)]$: Let $T_s = T(w_{k-1})$.
We extract $\langle a_k\rangle$ and $T(w_k)$ using two
{\sf rightpop}-operations on $\mathsf{DABA}_{\uparrow}$ (then $L_{\uparrow} = []$).
Then we construct the tree $T(a_k w_k b)$ using {\sf construct-prime}$(a_k, T(w_k), b)$.  Next, we 
construct the tree $T(w_{k-1} a_k w_k b)$ with 
{\sf concatenate}$(T_s, T(a_k w_k b))$. This tree becomes the new $T_s$.
\item $b \in \Sigma_r$ and $L_{\uparrow} = [\ldots, \langle a_{k-1} \rangle, T(w_{k-1}),\langle a_k \rangle, T(w_k)]$: We extract $T(w_{k-1})$, $\langle a_k\rangle$ and $T(w_k)$ using three 
{\sf rightpop}-operations on $\mathsf{DABA}_{\uparrow}$.
Then we construct the tree $T(a_k w_k b)$ using {\sf construct-prime}$(a_k, T(w_k), b)$.  Next, we 
construct the tree $T(w_{k-1} a_k w_k b)$ with 
{\sf concatenate}$(T(w_{k-1}), T(a_k w_k b))$.
This tree is added back to $L_{\uparrow}$ using a {\sf rightpush}-operation on $\mathsf{DABA}_{\uparrow}$.
\end{itemize}
Note that our data structure consisting of $L_{\downarrow}$, $T_s$ and $L_{\uparrow}$ is fully symmetric. 
Therefore, $\rpop$ and $\lpush$ can be implemented in an analogous way.

Note that the above data structure uses $\bigO(n \log n)$ bits: we have to store $\bigO(n)$ many pointers of bit length 
$\bigO(\log n)$. This concludes the proof of Theorem~\ref{thm-vpl}.
\end{proof}

\subsection{Deterministic 1-counter automata}

\label{sec:doca}

In this section we show that every deterministic 1-counter language has 
a $\tV$-algorithm with latency $\bigO(\log n)$ on a RAM with word size $\bigO(\log n)$.
We will use deterministic 1-counter automata (DOCAs) with $\varepsilon$-transitions. 
For this we use the definition from~\cite{BohmGJ13}:
A {\em deterministic 1-counter automaton}  is a tuple 
$\mathcal{A} = (Q_s, Q_r,\Sigma,\delta,  \pi,  \rho, q_0, F)$, where $Q_s$ is a finite set of stable states, 
$Q_r$ is a finite set of reset states ($Q_s \cap Q_r = \emptyset$), $\Sigma$ is a finite input alphabet, 
$\delta :  Q_s \times \Sigma \times \{0,1\} \to (Q_s \cup Q_r) \times \{-1,0,1\}$ is the transition function,
$\pi : Q_r \to \mathbb{N}$ maps every reset state $q$ to a period $\pi(q) > 0$, 
$\rho : \{ (q,k) \mid q \in Q_r, 0 \le k < \pi(q) \} \to Q_s$ is the reset mapping, $q_0$ is the initial state, and $F \subseteq Q$ is the set of final states.
It is required that if $\delta(q,a,i) = (q', j)$ then $i+j \geq 0$ to prevent the counter from becoming negative. Let $Q = Q_s \cup Q_r$.
The set of configurations of $\mathcal{A}$ is $Q \times \mathbb{N}$. For a configuration $(q,m)$ and $k \in \mathbb{N}$ we define $(q,m)+k = (q,m+k)$.
Intuitively, $\mathcal{A}$ reads an input letter whenever it is in a stable state and changes its configuration according to $\delta$.
If  $\mathcal{A}$ is in a reset state $q$ it resets the counter to zero and goes into a stable state that is determined (via the reset mapping $\rho$)
by $m \bmod \pi(q)$ when $m$ is the current counter value. Formally, we define 
the mappings
$\hat\delta : Q \times \mathbb{N}  \times \Sigma \to Q \times \mathbb{N}$ and $\hat\rho : Q \times \mathbb{N} \to Q_s \times \mathbb{N}$ as follows,
where $\sign : \mathbb{N} \to \{0,1\}$ is the signum function restricted to the natural numbers:
\begin{itemize}
\item If $q \in Q_r$ and $m \in \mathbb{N}$ then $\hat\rho(q,m) = (\rho(q, m \bmod \pi(q)), 0)$.
\item If $q \in Q_s$ and $m \in \mathbb{N}$ then $\hat\rho(q,m) = (q,m)$.
\item If $q \in Q_s$ and $m \in \mathbb{N}$ then $\hat\delta(q,m) = \hat\rho( \delta(q, a, \sign(m))+m )$.
\item If $q \in Q_r$ and $m \in \mathbb{N}$ then $\hat\delta(q,m) = \hat\delta(\hat\rho(q,m))$ (note that $\hat\rho(q,m) \in Q_s \times \{0\}$, for which 
$\hat\delta$ has been defined in the previous point).
\end{itemize}
We extend $\hat{\delta}$ to a function $\hat{\delta} : Q \times \mathbb{N} \times \Sigma^* \to Q \times \mathbb{N}$ in the usual way:
$\hat{\delta}(q,x,\varepsilon) = (q,x)$ and $\hat{\delta}(q,x,aw) = \hat{\delta}(\hat{\delta}(q,x,a),w)$. Then, 
$L(\mathcal{A}) = \{ w \in \Sigma^* \mid \hat{\delta}(q_0, 0, w) \in F \times \mathbb{N} \}$ is the language 
accepted by $\mathcal{A}$.
A language $L$ is a deterministic 1-counter language if $L = L(\mathcal{A})$ for some
deterministic 1-counter automaton $\mathcal{A}$.

Using the function $\hat\delta$ we can define also runs of $\mathcal{A}$ (we speak of $\mathcal{A}$-runs) on a word $w$ in the usual way:
For a configuration $(q,m)$ and a word $w = a_1 a_2 \cdots a_k$ the unique $\mathcal{A}$-run on the word $w$ starting in $(q,m)$ is the 
sequence of configurations $(q_0,m_0), (q_1,m_1), \ldots, (q_k,m_k)$ where 
$(q_i,m_i) = \hat\delta(q_0,0, a_1 \cdots a_i)$. We denote this run with $\mathsf{run}(q,m,w)$.

For a word $w \in \Sigma^*$ we define the {\em effect} $\hat\delta_w : Q \times \mathbb{N} \to Q \times \mathbb{N}$
by $\hat\delta_w(q,m) = \hat\delta(q,m,w)$.
It specifies how $w$ transforms configurations. 
It turns out that $\hat\delta_w$ can be completely reconstructed from restriction $\hat\delta_w \rest_{Q \times [0,|w|+p]}$
of $\hat\delta_w$ to the set $Q \times [0, |w|+p]$, where $p$ is the least common multiple 
of all $\pi(q)$ for $q \in Q_r$. 
To see this, assume that $(q,m)$ is a configuration with $m > |w|+p$. 
If in $\mathsf{run}(q,|w|+p,w)$ no state from $Q_r$ is visited and $\hat\delta_w(q, |w|+p) = (q', m')$, then we have
$\hat\delta_w(q, m) = (q', m' + m - |w| - p)$: basically, we obtain $\mathsf{run}(q,m,w)$
by shifting  $\mathsf{run}(q,|w|+p,w)$ upwards by $m- |w|-p$.
On the other hand, if a reset state that appears in $\mathsf{run}(q,|w|+p,w)$ then
$\hat\delta_w(q, m) = \hat\delta_w(q, |w|+i)$, where $i$ is any number in 
$[0,p]$ such that $|w|+i \equiv m \bmod p$.

In the following, we always assume that the effect $\hat\delta_w$ of a word $w$ is stored by
$\hat\delta_w \rest_{Q \times [0,|w|+p]}$, for which $\bigO(|w|)$ many registers of bit length $\bigO(\log |w|)$ suffice. 
 Given the effects $\hat\delta_u$ and $\hat\delta_v$ of two words $u,v$ of length at most $n$ 
 (stored in $\bigO(n)$ many registers of bit length $\bigO(\log n)$), we can
 compute the effect $\hat\delta_{uv}$ in time $\bigO(n)$ on a RAM with word size $\bigO(\log n)$.
 The computation of  $\hat\delta_{uv}$ on an argument from $Q \times [0, |uv|+p]$
 only involves simple arithmetic operations and can be done in constant time.
 
 \begin{theorem} \label{thm-DOCL}
 Every  deterministic 1-counter language $L$ has a $\tV$-algorithm with latency
 $\bigO(\log n)$, space complexity $\bigO(n \log^2 n)$ and word size $\bigO(\log n)$. 
 \end{theorem}

 \begin{proof}
 Let $w$ be the current window and $n$ its length and let $\mathcal{A}$ be a deterministic 1-counter automaton for $L$.
 Effects of words are always taken with respect to $\mathcal{A}$.
 Our $\tV$-algorithm will preserve the following invariants:
\begin{enumerate}
\item The current window $w$ is factorized into blocks $B_0,\ldots, B_m$
where $B_i$ has length $2^{a_i}$ and for some $k \in [-1, m]$ we have
$a_0  < a_1 < \cdots < a_k$ and $a_{k+1} > a_{k+2} > \cdots > a_m$.
A block of length $2^a$ is also called a level-$a$ block in the following. 
\item Every level-$a$ block $B$ is recursively factorized into two level-$(a-1)$ blocks that we call $B$'s left and right half.
So the blocks $B_0,\ldots, B_m$ are the maximal blocks, i.e., every other block is contained in some $B_i$.
The collection of all blocks is stored as a forest of full binary trees $T_0,\ldots, T_m$, where $T_i$ is the tree for block $B_i$.
\item For a certain time instant, let the age of a block $B$ be the number of window updates that have occurred between
the first point of time, where $B$ is completely contained in the window and the current point of time.
Note that $B$ can either enter the window on the left or on the right end.
If $B$ is a level-$a$ block and has age at least $2^a-1$, then the effect of $B$ must be completely computed
and stored in the tree node corresponding to block $B$.
We call such a block \emph{completed}. In particular, the effect of a
block of length $1$ has to be available after one further update.
\end{enumerate}
Later, when we explain how window updates are implemented we will see
how to preserve these invariants.

The following claim is an immediate consequence of the 3rd invariant.
\begin{claim}\label{claim-third-inv}
If a level-$a$ block $B$ has at least $2^a-1$ many symbols to its left as well as at least $2^a-1$ many symbols to its right in the window, then its age is at least $2^a-1$, so it must be completed.
\end{claim}
Thus, at every time instant, on every level $i$, there
can be at most $2$ non-completed blocks, namely the left most and right most block.

From the next claim it is easy to conclude that the query time is bounded by $\bigO(\log n)$:
\begin{claim}\label{claim-logblocks}
  At each time instant the window $w$ factors into $\bigO(\log n)$
  completed blocks.
\end{claim}

\begin{proof}[Proof of Claim~\ref{claim-logblocks}]
We show this for the case that
$a_k > a_{k+1}$ (the cases $a_k < a_{k+1}$ and $a_k = a_{k+1}$ can be treated analogously).
The middle block $B_k$ can be factorized into $2 a_k$ completed blocks, 
since if we consider the binary tree $T_k$ for $B_k$, then only the blocks on the left most and right most root-leaf path
of $T_k$  can be non-completed. This follows from invariant 3 and Claim~\ref{claim-third-inv} since each level-$b$ subblock of
$B_k$ that does not belong the left-most or right-most path in $T_k$ has another level-$b$ subblock to its left as well as to its right.
For the blocks $B_{k+1}, \ldots ,B_m$ we show that $B_{k+1} \ldots B_m$ can be factorized into at most
$a_{k+1} + m-k-1 \in \bigO(\log n)$ completed blocks, whereas $B_0 \cdots B_{k-1}$ can be factorized into at most
$a_{k-1} + k-1 \in \bigO(\log n)$ completed blocks. Consider the level-$a_i$ block $B_i$ for some $1 \le i \le k-1$. Every level-$b$ subblock of $B_i$ that does not belong
to the left most path in $T_i$ has at least $2^{b}-1$ many symbols to its right (namely the block $B_k$) as well as to its left (namely another level-$b$ subblock of  $B_i$), and is therefore
completed. But for blocks on the left most path in $T_i$ the same is true when they belong to some level $b \leq a_{i-1}$, because then the block
$B_{i-1}$ is to its left. Hence, there are at most $a_i - a_{i-1}$ non-completed blocks in $T_i$ and they form an initial part of the left most path.
Removing those blocks from $T_i$ leads to a factorization of $B_i$ into at most $a_i - a_{i-1} + 1$ completed blocks. Finally, for the first block $B_0$ we obtain with the 
same argument a factorization into at most $a_0$ completed blocks. By summing over all block $B_i$ ($0 \le i \le k-1$) we obtain a factorization
of $B_0 \cdots B_{k-1}$ into at most $a_0 + \sum_{1 \leq i \leq k-1} (a_i - a_{i-1} + 1) = a_{k-1} + k - 1$ completed blocks. For $B_{k+1} \cdots B_m$
we can argue analogously. The above arguments also show that it is easy to compute a factorization of the current window into completed blocks.
This factorization can be represented by a sequence of pointers to the tree nodes that correspond to the completed blocks in the factorization.
This proves Claim~\ref{claim-logblocks}.
\end{proof}
To check whether $w \in L$ we have to compute the configuration $\hat{\delta}_w(q_0, 0)$. For this take a factorization $w = C_1 C_1 \cdots C_{\ell}$ 
into $\ell \le \bigO(\log n)$ many completed blocks. Hence, the effect $\hat\delta_{C_i}$ is computed. We therefore apply the effects $\hat\delta_{C_1}, \ldots, \hat\delta_{C_\ell}$
in this order to the initial configuration $(q_0,0)$, which takes time $\bigO(\ell) \le \bigO(\log n)$.

It remains to describe how to deal with window updates and how to
preserve the invariants.
We first consider the operation $\lpush$. Let $t$ be the current point of time. We measure time in terms of window updates; every
window update increments time by 1.
Let us assume that $a_k \geq a_{k+1}$ (if $a_k < a_{k+1}$ then one has to replace $k$ by $k+1$ below).

Basically we make an increment on the binary representation of the number $2^{a_0} + \cdots +  2^{a_k}$. 
In other words, we consider the largest number $j \le k$ such that $a_i = i$ for all $0 \le i \le j$
and replace the blocks $B_0,\ldots,B_j$ together with the new $a$ in the window by a single level-$(j+1)$ block $B’_{j+1}$. 
Thereby also one new level-$i$ subblock
$B’_i$ of $B’_j$ for every $0 \le i \le j$ arises: we have $B'_0 = a$ and $B’_i = B’_{i-1} B_{i-1}$ for $1 \le i \le j+1$.
By invariant 3, all blocks $B_0,…,B_{j-1}$ are completed. 
If $j<k$ then $B_j$ is also completed, but if $j=k$ (and $a_k > a_{k+1}$) then we can only guarantee that 
the left half of $B_j$ is completed; its right half might be still non-completed.
Let us assume that $j=k$, which is the more difficult case. 

The effects of the new blocks $B’_0, B’_1,\ldots, B’_k$ can be computed bottom up as follows: The effect of $B'_0 = a$ is immediately
computed when $a$ arrives in the window. 
If the effect of $B'_{i-1}$ is already computed, then using the equality $B’_i = B’_{i-1} B_{i-1}$ and using the fact that $B_{i-1}$ is completed,
we can compute the effect of $B'_i$ in time $c \cdot 2^i$ for some constant $c$. In total, the computation of the effects
of all blocks $B’_0, B’_1,\ldots, B’_k$ needs time $\sum_{0 \le i \le k} c \cdot 2^i \le c \cdot 2^{k+1}$.
We amortize this work over the  next $2^k-1$ window updates by doing a constant amount of work in each step.
Thereby we ensure that at time instant $t+ 2^i - 1$, the new blocks $B’_0, \ldots, B'_i$ are completed for every $0 \leq i \leq k$.
In particular, at time instant $t+ 2^k - 1$, the block $B'_k$ is completed.
But at that time instant, also $B_k$ must be completed (if it is still in the window -- due to pops $B_k$ might haven been disappeared in the meantime) 
and we can reach the goal of completing $B’_{k+1}$ at time $t + 2^{k+1} - 1$ by still doing only a constant amount of work in each step.
This ensures that invariant 3 is preserved for every new block $B’_i$. Moreover,
before another level-$i$ block arises at the left end of the window (which can only happen every $2^i$ steps), the new block $B’_i$ is completed.
 Since the same arguments apply to $\rpush$, it follows that at every time instant, 
 on each level $i$ only two blocks are in the process of completion (one that was created on the left end and one that
was created on the right end). Since there are at most $\log(n)$ levels and for the completion of each block a constant amount of work is done in each step, we get the 
time bound $\bigO(\log n)$.

For $\lpop$, one has to remove all blocks along the leftmost path in the tree $T_0$ for block $B_0$,
which results in a sequence of smaller blocks of length $2^0, \ldots, 2^{a-1}$ if $B_0$ is a level $a$-block.
These blocks are already present in the tree $T_0$. 
Some of the blocks on the leftmost path of $T_0$ might be not completed so far. Of course we stop
the  computation of their effects. Invariant 1 is preserved, also in case $k = -1$ (where
$a_0 > a_1 > … > a_m$). 

Since the data structure is symmetric, the operations $\rpop$ and $\rpush$ can be implemented in the same way.
The algorithm uses space $\mathcal{O}(n \log^2 n)$: the dominating part are the values of the effect functions. On each level these are 
at most $n$ numbers of  bit length $\mathcal{O}(\log n)$. Moreover there at at most $\log n$ levels.
This concludes the proof.
 \end{proof}

\section{Open problems} \label{sec-open}

We conclude with some open problems:
In Theorem 3 we assume
that the size of the finite automaton for $L$ is a constant. One should also investigate 
how the optimal word size and latency depend on the number of states of the 
automaton. For space complexity, this dependency is investigated in \cite{GHKLM18}.

We showed that there is a real-time deterministic context-free language $L$ such that,
conditionally to the OMV conjecture, there is no $\oF$-algorithm for $L$ with logarithmic word
size and latency $n^{1/2-\epsilon}$ for any $\epsilon>0$. The best known upper bound in this setting for deterministic
context-free languages we are aware of is $\mathcal{O}(n/\log n)$. It is open, whether every deterministic
context-free language has a one-way fixed-size sliding window algorithm with logarithmic word
size and latency $n^{1-\epsilon}$ for some $\epsilon>0$.

For every deterministic one-counter language $L$, we showed that there is a $\tV$-algorithm with latency
 $\bigO(\log n)$ and word size $\bigO(\log n)$. Here, it remains open,
 whether the latency can be further reduced, maybe even to a constant. Also, our space bound $\bigO(n \log^2 n)$
 is not optimal. It would be nice to reduce it to $\bigO(n)$ without increasing the latency. The same problem appears
 for visibly pushdown languages, where our current space bound is $\bigO(n \log n)$ (with constant latency).
Finally, it would be interesting to see, whether the $\tV$-algorithm with latency $\bigO(1)$
for visibly pushdown languages can be extended to the larger
class of operator precedence languages \cite{Floyd63,Crespi-ReghizziM12}.

%\bibliographystyle{plainurl}
%\bibliography{fsttcs}

\bigskip

\end{document}